\documentclass[11pt]{article}
\usepackage{geometry}                
\usepackage{graphicx}
\usepackage{amsmath,amssymb, amsthm,epsfig,bm}
\usepackage{epstopdf}
\usepackage{setspace}
\usepackage{color,soul}
\sethlcolor{yellow}
\theoremstyle{plain} 
\newtheorem{theorem}{Theorem}
\newtheorem{lemma}[theorem]{Lemma}

\theoremstyle{definition}
\newtheorem{definition}{Definition}
\newtheorem{remark}{Remark}
\newtheorem{algorithm}{Algorithm}
\newtheorem{routine}{Routine}

\newcommand{\bfx}{\mathbf{x}} 
\newcommand{\bfy}{\mathbf{y}} 
\newcommand{\yobs}{\mathbf{y}_{obs}} 
\newcommand{\bfX}{\mathbf{X}}
\newcommand{\bfY}{\mathbf{Y}}
\newcommand{\bfW}{\mathbf{W}}
\newcommand{\bfN}{\mathbf{N}}
\newcommand{\bfV}{\mathbf{V}}
\newcommand{\bfA}{\mathbf{A}}
\newcommand{\bfH}{\mathbf{H}}
\newcommand{\bfF}{\mathbf{F}}
\newcommand{\bfg}{\mathbf{g}}
\newcommand{\calX}{\mathcal{X}}
\newcommand{\calB}{\mathcal{B}}
\newcommand{\calD}{\mathcal{D}}
\newcommand{\R}{\mathbb{R}}
\newcommand{\bmtheta}{\bm{\theta}}
\newcommand{\bmTheta}{\bm{\Theta}}
\newcommand{\bmnu}{\bm{\nu}}
\newcommand{\bmmu}{\bm{\mu}}
\newcommand{\bmomega}{\bm{\omega}}
\newcommand{\bmphi}{\bm{\phi}}

\newcommand{\I}{\bm{1}}
\newcommand{\defi}{\mathop{=}\limits^{\Delta}}  
\newcommand{\toas}{\mathop{\longrightarrow}\limits^{a.s.}}  
\newcommand{\dnorm}{\mathcal{N}}

\newcommand{\amax}{\mbox{argmax}}

\begin{document}

\title{{\bf Multi-Domain Sampling With Applications to Structural Inference of Bayesian Networks}
\thanks{To appear in {\it Journal of the American Statistical Association}.}
}
\author{Qing Zhou\thanks{
Qing Zhou is Assistant Professor, Department of Statistics, University of California, Los Angeles, CA 90095
(Email: zhou@stat.ucla.edu).
This work was supported in part by NSF grant DMS-0805491 and NSF CAREER Award DMS-1055286.
The author thanks the editor, the associate editor, and the two referees for helpful comments
and suggestions which significantly improved the manuscript. 
}}
\date{}
\maketitle

\begin{abstract}

When a posterior distribution has multiple modes, unconditional expectations, such as the posterior mean,
may not offer informative summaries of the distribution. Motivated by this problem,
we propose to decompose
the sample space of a multimodal distribution into domains of attraction of local modes.
Domain-based representations are defined to summarize the probability masses of and
conditional expectations on domains of attraction, which are much more informative
than the mean and other unconditional expectations. A computational method, the multi-domain sampler, 
is developed to construct domain-based representations for 
an arbitrary multimodal distribution. 
The multi-domain sampler is applied to structural learning of protein-signaling networks from 
high-throughput single-cell data, where a signaling network is modeled as a causal Bayesian network. 
Not only does our method provide a detailed landscape
of the posterior distribution but also improves the accuracy and the predictive power of
estimated networks.


\end{abstract}

\section{Introduction}

In Bayesian inference the information on an unknown parameter $\bmtheta$ given
an observed dataset $\yobs$ is contained in the posterior distribution $p(\bmtheta \mid \yobs)$.
When a posterior distribution does not belong to a well-characterized family of distributions,
Markov chain Monte Carlo (MCMC) is a standard computational approach
to Bayesian inference via sampling from the posterior distribution. Typical examples
of MCMC include the Metropolis-Hastings (MH) algorithm (Metropolis et al. 1953, Hastings 1970) 
and the Gibbs sampler (Geman and Geman 1984, Gelfand and Smith 1990, Tanner and Wong 1987).
Thorough reviews of recent developments on
Monte Carlo methods and their applications in Bayesian computation can be found in 
Chen et al. (2001) and Liu (2008).
The posterior mean $E(\bmtheta \mid \yobs)$ and other expectations are usually 
approximated from a Monte Carlo sample to summarize the posterior distribution. 
However, these unconditional expectations may not
offer good summaries of the information for Bayesian inference when a posterior distribution 
has multiple local modes. 
One can easily construct a multimodal posterior distribution of which the mean is located in
a low-density region and thus using it as an estimator for $\bmtheta$ lacks a conventional interpretation.
To extract more information contained in a multimodal posterior distribution,
it is desired to identify all major modes and calculate various statistics in appropriate
neighborhoods of these modes.

To achieve these tasks, we propose to partition the sample space of $\bmtheta$ 
into a collection of domains 
such that the posterior distribution restricted to each domain is unimodal. 
The most parsimonious partition that minimizes the number of domains is to use the
domains of attraction (to be defined rigorously later) of the local modes.
Take the trimodal distribution $p(\bmtheta)$ in Figure~\ref{fig:intro} as an illustration. The space
is partitioned into three domains, denoted by $\Theta_1,\Theta_2$ and $\Theta_3$: 
Each domain contains exactly one local mode;
if we move any $\bmtheta \in \Theta_k$ $(k=1,2,3)$ along the trajectory that always follows the gradient direction
of $p(\bmtheta)$, it will eventually reach the mode in the domain $\Theta_k$. 
We may then calculate various conditional expectations on these domains,
$E[h(\bmtheta)\mid \bmtheta \in \Theta_k]$ ($k=1,2,3$), for different functions $h$.
Together with the probability masses of the domains, $P(\bmtheta \in \Theta_k)$, 
they provide more informative summaries of the distribution $p(\bmtheta)$ than unconditional expectations. 
Such a summary is called a domain-based
representation (DR) for $p(\bmtheta)$.

\begin{figure}[t]
   \centering
   \includegraphics[width=2.5in]{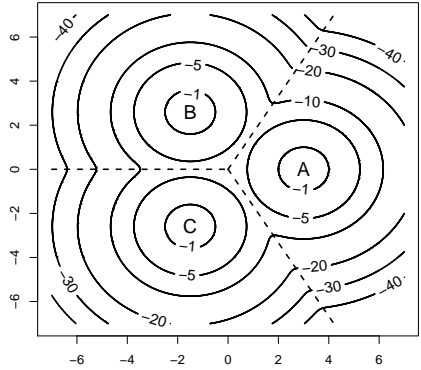} \\
   \caption{Contour plot of a two-dimensional density with three modes labeled A, B and C.
   The numbers report log densities of contours and the dashed curves mark the boundaries 
   between domains of attraction. \label{fig:intro}}
\end{figure}

Although desired, construction of DRs for an arbitrary distribution is very challenging in practice.
Sufficient Monte Carlo samples from domains of all local modes are necessary for estimating DRs, but
efficient sampling from a multimodal distribution has always been a hard problem.
In this article, we develop a computational method that is able to construct
domain-based representations for an arbitrary multimodal distribution. We partition the sample space
into domains of attraction and utilize an iterative weighting scheme aiming at
sampling from each domain with an equal frequency. The weighting scheme was proposed
by Wang and Landau (2001), and further developed and generalized by Liang (2005),
Liang et al. (2007) and Atchad\'e and Liu (2010) among others.
However, a direct application of these existing methods cannot provide
accurate estimation of DRs, due to at least two reasons. 
First, sample space partition used in these methods is usually predetermined
according to a set of selected density levels.
But partitioning the space into domains of attraction, as employed in our method, 
cannot be completed beforehand because it is a nontrivial job to detect all local modes and their
domains in real applications. Second, the above methods mostly rely on simple local moves
and lack a coherent global move to jump between different local modes.
To obtain accurate estimation of DRs, 
we propose a dynamic scheme to partition the sample space
into domains of attraction and devise a global move that utilizes
estimated DRs along sampling iterations to enable fast transitions between
multiple domains. Since the main feature of our method is to sample from
multiple domains and construct DRs, we call it the multi-domain (MD) sampler.

The MD sampler can be applied to a wide range of Bayesian inference problems
and it is particularly powerful in tackling problems with complicated posterior distributions.
Although there are many such applications in different fields, this article mainly 
concerns structural learning of causal Bayesian networks from experimental data. 
Learning network structure via Monte Carlo sampling is very
challenging as the multimodality of the posterior distribution is
extremely severe (Friedman and Koller 2003, Ellis and Wong 2008,
Liang and Zhang 2009). In this problem, a domain of attraction
is defined by a set of network structures, each represented
by a directed acyclic graph (DAG). 
In a sense, the goal of the MD sampler is to construct a detailed landscape
of the posterior distribution, which can provide new insights into the structural learning problem.
Application of our method to a scientific problem is illustrated by a study
on constructing protein-signaling networks from single-cell experimental data.
A living cell is highly responsive to its environment due to the existence of widespread and diverse signal transduction pathways.
These pathways constitute a complex signaling network as cross-talks usually exist between them. 
Knowledge of the structure of this network is critical to the understanding of various 
cellular behaviors and human diseases. Recent advances in biotechnology allow the biologist
to measure the states of a collection of molecules on a cell-by-cell basis. Such large-scale data
contain rich information for statistical inference of signaling networks, but powerful computational methods
are needed given the complexity in the likelihood function and the posterior distribution. 
With the MD sampler, not only can we build a signaling network from the posterior mean graph,
but also we may discover new pathway connections revealed by different domains
of the posterior distribution, which are not accessible by other approaches.

The remaining part of this article is organized as follows. 
Section~\ref{sec:doa} defines the domain of attraction and domain-based representation. 
In Section~\ref{sec:MD} we develop the MD sampler and its estimation of DRs,
with convergence and ergodicity of the sampler established in Appendix.
The method is tested in Section~\ref{sec:exRm} on an example in Euclidean space and
implemented in Section~\ref{sec:BN} for Bayesian inference of network structure with a simulation study. 
Section~\ref{sec:signaling} is 
the main application to the construction of signaling networks in human T cells. 
The article concludes with a discussion on related and future works.

\section{Domain-based representation}\label{sec:doa}

Let $p(\bfx)$, $\bfx \in \calX \subseteq \R^m$, be the density of the target distribution. Suppose that $p(\bfx)$
is differentiable and denote by $\nabla p(\bfx)$ the gradient of $p$ at $\bfx$. 
Define a differential equation
\begin{equation}
\frac{d\bfx(t)}{dt}= \nabla p(\bfx(t))
\end{equation}
and write a solution path of this equation as $\bfx(t), t\in [0,\infty)$, where $\bfx(0)$ is a chosen initial point.
Under some mild regularity conditions, $\bfx(\infty)$ converges
to a local mode of $p(\bfx)$, which is the basic intuition behind the gradient descent algorithm to maximize $p(\bfx)$.
Denote by $\{\bmnu_1,\ldots,\bmnu_K\}$ all the local modes, including the global mode, of $p(\bfx)$. 
For $\bfx\in \mathcal{X}$, let $\bfx(0)=\bfx$ and define the domain partition index by
\begin{equation}\label{eq:defI}
I(\bfx)=\left\{ 
\begin{array}{ll}
k, & \mbox{ if $\bfx(\infty) = \bmnu_k$, for $k=1,\ldots,K$} \\
0, & \mbox{ otherwise.}
\end{array}\right.
\end{equation}
It maps $\bfx$ to the index of the local mode to 
which the solution path starting at $\bfx$ converges.
\begin{definition}\label{def:Domain}
The domain of attraction of $\bmnu_k$ is
$D_k=\{\bfx \in \calX : I(\bfx)=k \}$ for $k=1,\ldots,K$.
For simplicity we may call $D_k$ an attraction domain or a domain of $p$.
\end{definition}
If the stationary points of $p(\bfx)$ have zero probability mass,
then $\{D_k:k=1,\ldots,K\}$ form a partition of the sample space $\calX$ except for 
a set of zero probability mass. This is the default setting for this article. 

Let $h(\bfx)$ be a $p$-integrable function of $\bfx$. 
Write the probability mass of $D_k$ and the conditional expectation of $h(\bfX)$ given $\bfX\in D_k$ as
\begin{eqnarray}
 \lambda_k  & = & P(\bfX\in D_k) = \int_{D_k} p(\bfx)  d\bfx, \label{eq:lambdak}\\
 \mu_{h,k} & = &E[h(\bfX)\mid \bfX\in D_k]=\frac{1}{\lambda_k} \int_{D_k} h(\bfx) p(\bfx) d\bfx, \label{eq:muhk}
\end{eqnarray}
respectively, for $k=1,\ldots,K$.
\begin{definition}
The domain-based representation of $h$ with respect to the distribution $p$ is a $2\times K$ array,
$DR_p(h)=\{(\mu_{h,k},\lambda_k): k=1,\ldots,K\}$.
\end{definition}
The DR is equivalent to 
the probability mass function of $E[h(\bfX) \mid I(\bfX)]$ that assigns probability $\lambda_k$ to 
$\mu_{h,k}$ for $k=1,\ldots,K$. 
It provides the expectation $E[h(\bfX)]=\sum_k  \lambda_k \mu_{h,k}$ and 
the decomposed contributions from the attraction domains of $p(\bfx)$. Such a representation
gives an informative low-dimensional summary of a multimodal distribution.
For a complex distribution with many local modes, however, we cannot afford to estimate 
$(\mu_{h,k},\lambda_k)$ for every domain when $K$ is too large and 
are less interested in domains of negligible
probability masses ($\lambda_k$ very close to zero).
Due to these reasons we define domain-based representations with respect to
a set of local modes $\{\bmnu_k: k = 1, \ldots, M\}$. 
Index all the local modes as $\bmnu_1,\ldots,\bmnu_M,\ldots,\bmnu_K$ $(M\leq K)$.
Define the domain partition index with respect to $\{\bmnu_k\}_{k=1}^M$ by $I_M(\bfx)=I(\bfx)$ if $1\leq I(\bfx) \leq M$
and $I_M(\bfx)=0$ otherwise. Then the sample space $\mathcal{X}$ can be partitioned
into $D_k=\{\bfx\in \mathcal{X} : I_M(\bfx)=k\}$ for $k=0,\ldots,M$, where $D_k$ is the domain of $\bmnu_k\, (k \geq 1)$ and 
$D_0 = \mathcal{X} - \bigcup_{k=1}^M D_k$. The DR of $h$ with respect to $\{\bmnu_k\}_{k=1}^M$
is defined by the array $\{(\mu_{h,k},\lambda_{k}): k=0,\ldots,M\}$, where $\lambda_0$ and
$\mu_{h,0}$ are defined for $D_0$ similarly as in Equations (\ref{eq:lambdak}) and (\ref{eq:muhk}), respectively.
Note that one can still obtain $E[h(\bfX)]=\sum_{k=0}^M  \lambda_k \mu_{h,k}$ after merging $D_{M+1},\ldots,D_K$
into $D_0$.

There is a geometric interpretation for the attraction domains of a posterior distribution.
Suppose that $\bfY=(\bfY_1,\ldots,\bfY_n)$ is a sample from an unknown distribution $\psi(\bfy)$.
We assume a parametric family $f_{\bmtheta}=f(\cdot \mid \bmtheta)$, $\bmtheta \in \Theta$, 
as the model for $\bfY$ and put a prior $\pi(\bmtheta)$
on the unknown parameter $\bmtheta$. The posterior distribution of $\bmtheta$ is given by
\begin{equation*}
p(\bmtheta \mid \bfY)  \propto  \pi(\bmtheta) \prod_{i=1}^n f(\bfY_i\mid \bmtheta) 
 \approx  e^{ n E [\log f(\bfY_1 \mid \bmtheta)]  }
\end{equation*}
when the sample size $n$ is large, where 
$E [\log f(\bfY_1 \mid \bmtheta)] = \int \log [f(\bfy\mid\bmtheta)] \psi (\bfy) d\bfy$.
Denote the Kullback-Leibler (KL) divergence between $\psi$ and $f_{\bmtheta}$ by
\begin{equation*}
d_{KL}(\psi \| f_{\bmtheta})=\int \log\left[ \frac{\psi(\bfy)}{f(\bfy\mid \bmtheta)} \right] \psi(\bfy) d\bfy. 
\end{equation*}
Then, $p(\bmtheta \mid \bfY)  \propto   \exp[ - n d_{KL}(\psi \| f_{\bmtheta})]$ when $n$
is large. In the space of density functions, we can regard $d_{KL}(\psi \| f_{\bmtheta})$ as the ``distance" to 
the point $\psi$ from a point in the manifold $\mathcal{M} = \{f_{\bmtheta} : \bmtheta \in \Theta \}$. 
Note that $\psi$ is not necessarily in $\mathcal{M}$ if our model assumption on $\bfY$ is incorrect.
Then $p(\bmtheta\mid \bfY)$ may be interpreted as a Boltzmann distribution 
on the manifold $\mathcal{M}$ under a potential field $n d_{KL}(\psi \| f_{\bmtheta})$. 
This potential pushes every $f_{\bmtheta}$, indexed by $\bmtheta$,
towards $\psi$, and the collection of $\bmtheta$
which will be driven to an identical
stationary point in $\mathcal{M}$ forms an attraction domain of $p(\bmtheta \mid \bfY)$.

\section{The multi-domain sampler} \label{sec:MD}

To develop an algorithm that is able to construct domain-based
representations with respect to the target distribution $p$, it is necessary 
to identify the attraction domain of any $\bfx\in \mathcal{X}$. When $p(\bfx)$ is differentiable,
this can be achieved by application of the gradient descent (GD) algorithm that finds local modes 
of $p(\bfx)$, or $\log p(\bfx)$ for computational convenience. 
Generalization to the space of network structures will be discussed later. 
 
A naive two-step approach to the construction of DRs is quite obvious. We may first apply a Monte Carlo algorithm
to simulate a sample $\{\bfX^t\}_{t=1}^n$ from $p(\bfx)$ or from a diffuse version of $p(\bfx)$, 
e.g., $[p(\bfx)]^{1/\tau}$ for $\tau > 1$ as used in parallel tempering (Geyer 1991).
Then, for every $t$ we determine $I(\bfX^t)$ by a GD search 
initiated at $\bfX^t$ to find to which domain it belongs. 
This approach partitions the sample
into attraction domains so that we can estimate the probability masses and conditional
expectations for all identified domains. Although simple to implement, this two-step approach has a few drawbacks
in terms of efficiency. Without a careful and specific design the Monte Carlo algorithm, even targeting at a diffuse
version of $p(\bfx)$, may not generate 
enough samples from all major domains or may completely miss some modes.
As a result, estimation on some attraction domains
may be inaccurate or unavailable. In addition, this approach does not utilize the information on the target distribution
provided by the constructed DRs. To overcome these drawbacks, we
develop the MD sampler that may achieve simultaneously 
an efficient simulation from a multimodal distribution and an accurate construction of
domain-based representations, with comparable computational complexity as the naive 
two-step approach. 

\subsection{The main algorithm}

We wish to sample sufficiently from the majority of attraction domains. 
However, the density at the boundary between two neighboring domains is often exponentially
low (e.g., Figure~\ref{fig:intro}), 
which makes it difficult for an MH algorithm or a Gibbs sampler to jump across multiple domains.
Thus, we need to allow the sampler to generate enough samples from such low-density regions
that connect different domains. These considerations motivate 
the following double-partitioning design in the MD sampler.

Suppose that we are given a set of local modes of $p(\bfx)$, $\{\bmnu_1,\ldots,\bmnu_M\}$, which may include the 
global mode. Given $\infty=H_0>H_1>\ldots >H_L=-\infty$, define the density partition index $J(\bfx)=j$ 
if $\log p(\bfx)  \in [H_j, H_{j-1})$ for $j=1,\ldots,L$.
We partition the space $\mathcal{X}$ into $(M+1) \times L$ subregions, 
\begin{equation}\label{eq:aepar}
D_{kj}= \{ \bfx\in \mathcal{X} : I_M(\bfx)= k, J(\bfx)=j \},\, k=0,\ldots,M, j=1,\ldots,L,
\end{equation}
where $I_M(\bfx)$ is the domain partition index with respect to $\{\bmnu_k\}_{k=1}^M$.
Then, the attraction domain of the local mode $\bmnu_k$ is $D_k=\bigcup_j D_{kj}\,(1\leq k\leq M)$. Note that 
some $D_{kj}$ may be empty; if $\log p(\bmnu_k)<H_i$ then all $D_{kj}$ for $j\leq i$ are empty.
In what follows, we only consider nonempty subregions.
For a given matrix $\bfW = (w_{kj})_{(M+1) \times L}$, define a working density 
\begin{equation}\label{eq:weightds}
p(\bfx; \bfW) \propto \sum_{k=0}^{M} \sum_{j=1}^L \frac{p(\bfx) \I (\bfx \in D_{kj} )}{\exp(w_{kj})},
\end{equation}
where $\bm{1}(\cdot)$ is an indicator function. Let $\bfW^*=(w_{kj}^*)$ such that 
$\exp(w^*_{kj}) = \int_{D_{kj}} p(\bfx) d\bfx$.
Then, the probability masses of $D_{kj}$ are identical under $p(\bfx;\bfW^*)$.
Sampling from $p(\bfx;\bfW^*)$ has two immediate implications. 
First, the sample sizes on the attraction domains, $\{D_k\}_{k=1}^M$, will be comparable, 
and thus, domain-based representations can be constructed with a high accuracy. 
Note that commonly used MCMC strategies for multimodal distributions, such as tempering,
cannot generate samples of comparable sizes from different domains.
Second, the sampler will stay in low-density
regions (e.g., $D_{kL}$) for a substantial fraction of time, which makes it practically possible to jump between domains.
Conversely, domain-based representations may be utilized to design efficient local and 
global moves for sampling from $p(\bfx;\bfW^*)$. We may
construct online estimate of the covariance matrix on the domain of a local mode, which can be
used for tuning the step size of a local move in this domain. 
For a multimodal distribution, tuning step size for each domain is more useful than 
tuning the overall step size (Harrio et al. 2001).
Once we have identified sufficient local modes and estimated covariances of their respective domains, 
we can use them to design global moves that may jump from one domain to another. 
As one can see, these proposals can be implemented only if we have partitioned samples into attraction domains.

For $\bfx,\bfy \in \mathcal{X}$, let $q(\bfx,\bfy)$ 
be a proposal from $\bfx$ to $\bfy$ and $r(\bfx; \bmtheta,\bfV)$ a distribution with 
parameters $\bmtheta$ and $\bfV\in \mathcal{V}$. We first develop the main algorithm
of the MD sampler, assuming that the density ladder $\{H_j\}$ is fixed and $M$ local modes
$\{\bmnu_k\}_{k=1}^M$ have been identified. 
Dynamic update of these parameters will be discussed in Section~\ref{sec:burnin} as the burn-in algorithm.
Let $\bfg_k(\bfx)$ be a map from $\calX$ to $\mathcal{V}$ for $k=1,\ldots,M$.  
\begin{algorithm}[The main algorithm]\label{alg:main}
Initialize $\bfW^1=(w^1_{kj})$, parameters  
$\bfV_k^1\in \mathcal{V}$ $(k=1,\ldots,M)$, $p_{mx} \in [0,1)$, and $\bfX^1 \in \mathcal{X}$. Set $\gamma_1\leq 1$.
For $t=1,\ldots,n$:
\begin{enumerate}
\item{Draw $\bfY$ from $q(\bfX^t,\bfy)$ with probability $(1-p_{mx})$ or from the mixture distribution 
$\frac{1}{M} \sum_{k=1}^M r(\bfy ; \bmnu_k, \bfV_k^t)$ with probability $p_{mx}$.}
\item{Determine the density partition index $J(\bfY)$ and perform a GD search initiated at $\bfY$
to determine the domain partition index $I_M(\bfY)$.}
\item{Accept or reject $\bfY$ according to the MH ratio targeting at $p(\bfx;\bfW^t)$ to obtain $\bfX^{t+1}$.}
\item{For $k=0,\ldots,M$ and $j=1,\ldots,L$ update 
\begin{equation}\label{eq:updatew}
w^{t+1}_{kj} =  w^{t}_{kj} + \gamma_t \I(I_M(\bfX^{t+1})=k,J(\bfX^{t+1})=j),
\end{equation}
for $k=1,\ldots,M$ update
\begin{equation}\label{eq:updateV}
\bfV_k^{t+1}  =  \bfV_k^t + \frac{\gamma_t}{2} \left[\bfg_k(\bfX^{t+1}) -\bfV_k^t \right] \I(I_M(\bfX^{t+1})=k),
\end{equation}
and determine $\gamma_{t+1}$.
}
\end{enumerate}
\end{algorithm}

We may regard $w^t_{kj}$ as a weight for the subregion $D_{kj}$. 
After each visit to $D_{kj}$, the weight $w_{kj}^{t+1}$ increases by $\gamma_t$ unit (\ref{eq:updatew}),
which decreases the probability mass of $D_{kj}$ under the working density 
$p(\bfx; \bfW^{t+1})$ (\ref{eq:weightds}). Such a weighting scheme aims to visit every $D_{kj}$ uniformly. 
There are two typical choices of $\{\gamma_t\}$. The first choice follows the standard design in 
stochastic approximation which employs a predetermined sequence such that 
$\sum_{t=1}^{\infty} \gamma_t = \infty$
and $\sum_{t=1}^{\infty} \gamma_t^{\zeta} < \infty$ for $\zeta\in (1,2)$
(Andrieu et al. 2005, Andrieu and Moulines 2006, Liang et al. 2007).
The second design, originally proposed by Wang and Landau (2001), adjusts $\gamma_t$ in 
an adaptive way. However, there is difficulty in establishing its convergence (Atchad\'e and Liu 2010),
and thus we adopt a modified Wang-Landau (MWL) design to update $\{\gamma_t\}$ in the MD sampler.
Initialize $c^1_{kj}=0$ for all $k$ and $j$ in Algorithm~\ref{alg:main}. The MWL
update at iteration $t\,(t=1,\ldots,n)$ is given below.
\begin{routine}[MWL update]\label{MWLdesign}
If $\gamma_t<\epsilon_{\gamma}$, set $\gamma_{t+1} =\gamma_t (\gamma_t +1)^{-1}$; otherwise:
\begin{itemize}
\item{Set $c^{t+1}_{kj}=c^{t}_{kj}+\I(\bfX^{t+1} \in D_{kj})$ for all $k,j$ and 
calculate $\Delta c^{t+1}_{\max} = \max_{k,j} | c^{t+1}_{kj} - \bar{c}^{\,t+1} |$,
where $\bar{c}^{\,t+1}$ is the average of all $c^{t+1}_{kj}$.}
\item{ 
If $\Delta c^{t+1}_{\max}\geq\eta \bar{c}^{\,t+1}$, set $\gamma_{t+1}=\gamma_t$;
otherwise, set $\gamma_{t+1}=\rho \gamma_t$ and $c^{t+1}_{kj}=0$ for all $k,j$.
}
\end{itemize}
\end{routine}
That is, if $\gamma_t \geq \epsilon_{\gamma}$ we decrease $\gamma_t$ ($\rho<1$) 
only when the sampler has visited every subregion $D_{kj}$
with a roughly equal frequency since our last modification of $\gamma_t$. 
Let $t_c=\min\{t:\gamma_t<\epsilon_{\gamma}\}$.
For $t>t_c$ the update becomes deterministic with $\gamma_t = 1/(t+\xi)$,
where $\xi=\gamma_{t_c}^{-1}-t_c$.
The default setting for all the results in this article is given by 
$\rho=0.5$, $\eta=0.25$ and $\epsilon_{\gamma}=10^{-4}$.

Under some regularity conditions and the MWL update of $\{\gamma_t\}$, 
\begin{equation}\label{eq:wlimit}
\exp(w_{kj}^{n}) \toas \int_{D_{kj}} p(\bfx) d\bfx = \exp(w^*_{kj}),
\end{equation}
after being normalized to sum up to one, 
\begin{equation}\label{eq:Vlimit}
\bfV_k^n \toas \frac{\int_{D_{k}} \bfg_k(\bfx)  p(\bfx;\bfW^*) d\bfx}{\int_{D_k}p(\bfx;\bfW^*)d\bfx} \defi \bfV_k^*,
\end{equation}
\begin{equation}\label{eq:flimit}
\frac{1}{n}\sum_{t=1}^n h(\bfX^t) \toas \int_{\calX}  h(\bfx) p(\bfx;\bfW^*) d\bfx, 
\end{equation}
as $n\to\infty$. See Theorem~\ref{thm:convergence} in Appendix for more details. If $\mathcal{X} \subseteq \R^m$, 
we often choose
$\bfg_k(\bfx)= (\bfx-\bmnu_k)(\bfx-\bmnu_k)^{\mathsf{T}}$
so that $\bfV_k^*$ is close to the covariance matrix of the conditional distribution 
$[\bfX\mid \bfX\in D_k]$, where $\bfX\sim p(\bfx;\bfW^*)$.
We use the mode $\bmnu_k$
instead of the mean because the mode can be obtained accurately via a GD algorithm.
The use of $\gamma_t/2\, (<1)$ in Equation \eqref{eq:updateV} ensures that 
$\bfV_k^{t+1}$ is positive definite if $\bfV_k^t$ is positive definite.

There are two types of proposals in step 1 of the algorithm, a local proposal
$q(\bfX^t,\bfy)$ and a mixture distribution proposal.
One advantage of partitioning samples into attraction domains is embodied in the mixture distribution
proposal, in which we randomly draw a domain partition index $k \in \{1,\ldots,M\}$ and then
propose a sample $\bfY$ from $r(\bfy; \bmnu_k, \bfV_k^t)$. 
Equal mixture proportions ($1/M$) are used because a uniform sampling across domains is preferred.
The default choice of the distribution $r(\bfy; \bmnu_k, \bfV_k^t)$
in $\R^m$ is $\dnorm(\bmnu_k,\bfV_k^t)$ for $k=1,\ldots,M$, which gives a mixture normal proposal that matches the
mode and the covariance on each domain of the working target $p(\bfx;\bfW^t)$.
This proposal uses a mixture distribution to approximate the multimodal target. 
It can generate efficient global jumps from one domain to another if $p(\bfx;\bfW^t)$
on the domain $D_k$ can be well approximated by $r(\bfx; \bmnu_k, \bfV_k^t)$ with the identified mode $\bmnu_k$ and 
the estimated $\bfV_k^t$. For simplicity we call this proposal the \emph{mixed jump}.
The typical design for $q(\bfX^t,\bfy)$ in $\R^m$ is to proposal $\bfY\sim \dnorm(\bfX^t,  \sigma^2 \mathbf{I})$, where
$\sigma^2$ is a scalar and $\mathbf{I}$ is the identity matrix. 
However, when the covariances are very different between domains,
using a single local proposal may cause high autocorrelation, since the step
size might be either too big for domains with small covariances or too small for those with
large covariances, or both. In this case, we may incorporate an adaptive local proposal,
$\bfY\sim \dnorm(\bfX^t,  \sigma^2 \bfV_{I_M(\bfX^t)}^t)$, in addition to $q(\bfX^t,\bfy)$,
such that the learned covariance structure of a domain is utilized to guide the local proposal.
This shows another advantage of the domain-partitioning design.

\begin{remark} 
We summarize the unique features of the main algorithm.
First, domain partitioning is incorporated in the framework of the Wang-Landau (WL) algorithm. 
This allows a more uniform sampling from different domains, which facilitates
construction of DRs. At each iteration, a GD search is employed
to determine $I_M(\bfY)$ and thus the computational complexity of this algorithm
is comparable to the naive two-step approach. Second, an adaptive global move, the mixed jump, is proposed given
DRs constructed along the iteration, which utilizes identified modes and learned covariances
to achieve between-domain moves.
\end{remark}

\begin{remark}\label{rmk:diagnostics}
Verification of the regularity conditions for convergence of the algorithm (see Appendix)
is recommended before application. Furthermore, we suggest a few convergence diagnostics that can 
be conveniently used in practice. First, $\gamma_n$ should be small enough at the final iteration
and the frequency of visiting different $D_{kj}$ should be roughly identical. Second, 
$\bfW^n$ and $\bfV_k^n$ should have converged with an acceptable accuracy. 
Violations of these two criteria indicate that more iterations may be necessary.
Third, the adaptive parameters used in the mixed jump ($\bfV_k^t$) should always stay in a reasonable
range. For example, if $\bfV_k^t$ is a covariance matrix, one may check whether its eigenvalues
are close to zero or unreasonably large, which may indicate divergence of the current run. 
If the last criterion is not satisfied, it is suggested to reinitialize the MD sampler 
with a smaller $\gamma_1$.
\end{remark}
  
\subsection{The burn-in algorithm}\label{sec:burnin}

In practical applications of the MD sampler, the density ladder $\{H_j\}$ and the local modes $\{\bmnu_k\}$ are
updated dynamically in a burn-in period before the main algorithm (Algorithm \ref{alg:main}).
The dynamic updating schemes are crucial steps for constructing domain-based
representations in real applications, as one cannot partition the sample space into domains of attraction beforehand.
We set $H_1,\ldots,H_{L-1}$ as an evenly spaced sequence so that $\Delta H= H_j-H_{j+1}$
is a constant for $j=1,\ldots,L-2$. 
Let $\{H_j^t\}$ be the density ladder and $\Lambda^t=\{ \bmnu_k^t : k=1,\ldots, M^t \}$ 
be the set of $M^t$ identified modes at iteration $t$.
Let $K^*$ be the maximum number of modes to be recorded and
denote by $\bmnu_{(\bfx)}$ the mode of the domain that $\bfx$ belongs to.
Let $\mathbf{0}$ be the zero matrix with dimension determined by the context.
The following routine is used to update $\Lambda^t$ when a new sample $\bfY$ is proposed.
\begin{routine} \label{rt:updatemodes}
Let $s^t=\underset{1\leq k \leq M^t}{\mbox{argmin}} \;p(\bmnu_k^t)$.  
\begin{itemize}
\item{If $\bmnu_{(\bfY)} \not\in \Lambda^t$ and $M^t<K^*$, set $\Lambda^{t+1}=\Lambda^t \cup \{\bmnu_{(\bfY)}\}$,
$M^{t+1}=M^t+1$, and initialize $w^t_{M^{t+1}j}=0$ for all $j$;}
\item{if $\bmnu_{(\bfY)} \not\in \Lambda^t$, $M^t=K^*$ and $p(\bmnu_{(\bfY)})>p(\bmnu^t_{s^t})$,
set $\bmnu_{s^t}^{t+1}=\bmnu_{(\bfY)}$, $\bmnu_k^{t+1}=\bmnu_k^t$ for $k\ne s^t$, $M^{t+1}=M^t$, and
assign $w_{0j}^t \Leftarrow w_{0j}^t+ w_{s^tj}^t$ and $w_{s^tj}^{t}\Leftarrow 0$ for all $j$.}
\item{otherwise set $\Lambda^{t+1}=\Lambda^t$ and $M^{t+1}=M^t$.}
\end{itemize}
\end{routine}
According to this routine, we record at most the $K^*$ highest modes identified during the burn-in period. 
If there are more than $K^*$ modes, Algorithm~\ref{alg:main} will construct DRs with respect to the 
recorded modes. The weights $(w_{kj}^t)$ are updated when a new mode replaces an old one in $\Lambda^{t}$,
for which the assignment operator `$\Leftarrow$' is used to distinguish from equality.

The density ladder $\{H_j^t\}$ is adjusted such that
$H_1^t$, the lower bound of the highest density partition interval, is close to $\log u^*$, 
where $u^*$ is the density of the highest mode identified so far. 
If $\log u^* > H^t_1+\Delta H$ 
we move upwards the density ladder by $\Delta H$ unit and update the weights $(w^t_{kj})$ accordingly,
with details provided in Routine~\ref{rt:updateladder}.
This strategy helps the sampler to explore the high-density part,
which is important for statistical estimation and finding the global mode.
\begin{routine} \label{rt:updateladder}
Given $\Lambda^{t+1}$, find $u^{t+1}=\max\{p(\bmnu_k^{t+1}):  k=1, \ldots, M^{t+1}\}$.
\begin{itemize}
\item{If $\log u^{t+1}> H_1^t+\Delta H$,
set $H_j^{t+1}=H_j^t+\Delta H$ for $j=1,\ldots,L-1$; 
for $k=0,\ldots,M^{t+1}$, assign $w_{kL}^{t}\Leftarrow w_{k(L-1)}^{t}+w_{kL}^{t}$, $w_{kj}^{t}\Leftarrow w_{k(j-1)}^{t}$ for
$j=L-1,\ldots,2$, and $w^t_{k1}\Leftarrow 0$;}
\item{otherwise set $\{H_j^{t+1}\}= \{H_j^{t}\}$.}
\end{itemize}
\end{routine}

\begin{algorithm}[The burn-in algorithm]\label{alg:burnin}
Input $L,\Delta H$ and $K^*$. Set $\gamma_0=1$. 
Initialize $\bfX^1\in \mathcal{X}$, $\Lambda^1=\{\bmnu_{(\bfX^1)}\}$, $M^1=1$, $\bfW^1=(w_{kj}^1)_{2\times L}=\mathbf{0}$,
and $\bfV_1^1$. Set $H_1^1=\log p(\bmnu_{(\bfX^1)})$ and
$H_{j}^1=H_{j-1}^1-\Delta H$ for $j=2,\ldots,L-1$.
For $t=1,\ldots,B$:
\begin{enumerate}
\item{Draw $\bfY\sim q(\bfX^t,\bfy)$ and find $\bmnu_{(\bfY)}$ by a GD search.}
\item{Given $\bmnu_{(\bfY)}$, update $\Lambda^{t+1}$ and $M^{t+1}$ by
Routine \ref{rt:updatemodes}; if $\bmnu_{\ell}^{t+1}=\bmnu_{(\bfY)}$ is a new mode in $\Lambda^{t+1}$, 
initialize $\bfV_{\ell}^{t}$. Given $\Lambda^{t+1}$, update $\{H_j^{t+1}\}$ by Routine \ref{rt:updateladder}.}
\item{Given $\Lambda^{t+1}$ and $\{H_j^{t+1}\}$,
accept or reject $\bfY$ with the MH ratio targeting at $p(\bfx;\bfW^t)$ to obtain $\bfX^{t+1}$.
}
\item{Execute step 4 of Algorithm \ref{alg:main} with $\gamma_t=\gamma_0$.}
\end{enumerate}
\end{algorithm}

\begin{remark}
Note that $\gamma_t=1$ for every iteration in the burn-in algorithm. This pushes the sampler
to explore different regions in the sample space so that more local modes can be identified. In this case,
the weight $w^t_{kj}$ records the number of visits to $D_{kj}$ before the $t$th iteration, 
which is the reason for our updating
schemes on $\{w^t_{kj}\}$ in Routine~\ref{rt:updatemodes} when a mode is updated in $\Lambda^{t+1}$
and in Routine \ref{rt:updateladder} when the density ladder changes.
\end{remark}

\begin{remark}\label{rmk:burninopt}
The burn-in algorithm can be used as an optimization method that searches for up to $K^*$
local modes of the highest densities. As demonstrated in the Bayesian network applications, 
this algorithm is very powerful in finding global modes.
\end{remark}

The MD sampler requires only a few input parameters, $L,\Delta H, p_{mx}$, and $K^*$. A practical rule is
to choose $L$ and $\Delta H$ such that the range of the density partition intervals, $L\Delta H$ in log scale, 
is wide enough to cover important regions. 
In this paper, we set $L\Delta H$ around 20 for the low-dimensional test example in Section~\ref{sec:exRm}
and around 200 for learning Bayesian networks in Sections~\ref{sec:BN} and \ref{sec:signaling}. 
By default the probability of proposing a mixed jump $p_{mx} =0.1$.
The effect of keeping only $K^*$ modes will be studied later with the examples. 

\subsection{Statistical estimation}

The domain-based representation of $h$ is constructed by estimating $\lambda_k$ (\ref{eq:lambdak}) and 
$\mu_{h,k}$ (\ref{eq:muhk}) for $k=0,\ldots,M$ with post burn-in samples, denoted by $\{\bfX^{t+1}\}_{t=1}^n$.
Let $k^t = I_M(\bfX^{t+1})$, $j^t = J(\bfX^{t+1})$, $a_t=\sum_{k,j}\exp(w^t_{kj})$, and
$\exp(\tilde{w}^t_{kj})= \exp({w}^t_{kj})/a_t$ such that $\sum_{k,j}\exp(\tilde{w}^t_{kj})=1$.
The key identity for our estimation is 
\begin{equation}\label{eq:iwest}
\frac{\sum_{t=1}^{\infty} h(\bfX^{t+1}) \exp(\tilde{w}^t_{k^tj^t}) }{\sum_{t=1}^{\infty} \exp(\tilde{w}^t_{k^tj^t})} 
\toas \int_{\calX} h(\bfx) p(\bfx) d\bfx, 
\end{equation}  
which follows from \eqref{eq:flimit} as 
$\exp(\tilde{w}^t_{k^tj^t})\toas \exp(w^*_{k^tj^t}) \propto p(\bfX^{t+1})/p(\bfX^{t+1}; \bfW^*)$ asymptotically \eqref{eq:wlimit}.
See Liang (2009) and Atchad\'e and Liu (2010) for similar results. However, $\tilde{\bfW}^t=(\tilde{w}_{kj}^t)$ may
be far from $\bfW^*$ even for post burn-in iterations. Thus, it is desired to use a weighted version
of \eqref{eq:iwest} so that $\bfX^{t+1}$ will carry a higher weight if $\tilde{\bfW}^t$ is closer to $\bfW^*$.
Since decrease in $\gamma_t$ indicates convergence of the MD sampler and 
$a_{t+1}/a_t \toas e^{c\gamma_t}$ for $c\in(0,1)$ (supplementary document), 
a reasonable choice is to weight $\bfX^{t+1}$ by $a_t$ so that
unnormalized $(w^t_{kj})$ will be used in \eqref{eq:iwest}.
Consequently, $DR_p(h)$ is constructed with
\begin{eqnarray*}
\hat{\lambda}_{k} & = & \frac{\sum_{t=1}^n \bm{1}(\bfX^{t+1} \in D_{k})\exp(w^t_{k^tj^t})}{\sum_{t=1}^n \exp(w^t_{k^tj^t})}, \\
\hat{\mu}_{h,k} & = &  \frac{\sum_{t=1}^n h(\bfX^{t+1}) \bm{1}(\bfX^{t+1} \in D_{k})\exp(w^t_{k^tj^t})}{\sum_{t=1}^n \bm{1}(\bfX^{t+1} \in D_{k})
\exp(w^t_{k^tj^t})},
\end{eqnarray*}
for $k=0,\ldots,M$. Then, $\mu_h = E[h(\bfX)]$ is estimated by $ \hat{\mu}_h = \sum_k \hat{\lambda}_k\hat{\mu}_{h,k}$.
Please see supplementary document for more discussion on this weighted estimation.

In the next three sections we demonstrate the effectiveness of the MD sampler in statistical estimation,
especially estimation of DRs, compared to the naive two-step approach.
For all examples, we employ the WL algorithm with the MWL update (Routine~\ref{MWLdesign}) as 
the Monte Carlo method in the two-step approach.
To minimize hidden artifacts in a comparison due to coding differences, 
we implement the WL algorithm with the same burn-in and main algorithms of the MD sampler.
In the main algorithm (Algorithm \ref{alg:main}) we replace the updating scheme in Equation (\ref{eq:updatew}) with 
\begin{equation}\label{eq:WLupdatew}
w^{t+1}_{kj} =  w^{t}_{kj} + \gamma_t \I(J(\bfX^{t+1})=j)
\end{equation}
for $k=0,\ldots,M$ and $j=1,\ldots,L$ and modify the burn-in algorithm accordingly, so that
$w^t_{0j}=\cdots=w^t_{Mj}\defi w^t_j$ for every iteration. Consequently,
the working density is effectively
\begin{equation}\label{eq:WLworkingtarget}
p(\bfx; \bfW^t) \propto \sum_{j=1}^L \frac{p(\bfx) \I (J(\bfx)=j)}{\exp(w^t_{j})}
\end{equation}
as used in the WL algorithm, which is a diffuse version of $p(\bfx)$ such that each density partition interval
will be equally sampled after convergence. Note that the same GD search is applied 
at each iteration to partition samples into attraction domains for estimating DRs.
Our comparison aims to highlight the effect of domain partitioning and the mixed jump
in the MD sampler which are the key differences from the WL algorithm.

\section{A test example}\label{sec:exRm}

We test the MD sampler with an example in $\R^m$. For this example, 
domain-based representations can be obtained via one-dimensional
numerical integration with a high accuracy, which provides the basis to evaluate our estimation. 
We choose $K^*=100$, which is greater than the total number
of local modes, to construct complete DRs. 

Let $\bfx=(x_1,\ldots,x_m)$. The Rastrigin function (Gordon and Whitley 1993) 
is defined as
\begin{equation} \label{eq:Rastrigin}
R(\bfx)=\sum_{i=1}^{m}{x_i^2}+A\left[m-\sum_{i=1}^{m}\cos(\pi x_i)\right],
\end{equation}
where $A$ is a positive constant.
We set $A=2$ and $m=4$ in (\ref{eq:Rastrigin}) to obtain our target distribution 
$p(\bfx) \propto \exp[-R(\bfx)]$, which has
$3^4=81$ local modes formed by all the elements of the product set $\{-1.805,0,1.805\}^4$. These local modes
have five distinct log density values, 0, $-3.62$, $-7.24$, $-10.87$, and $-14.49$, 
dependent on the combinations of their coordinates.
They are grouped accordingly into five layers so that the number of zeros and 
the number of $\pm 1.805$ in the coordinates 
of a local mode at the $k$th layer are $(5-k)$ and $(k-1)$, respectively, for $k=1,\ldots,5$.
The attraction domains of local modes at the same layer have identical probability masses 
and identical conditional means up to a permutation and change of signs of the coordinates.

We applied the MD sampler 100 times independently, each run
with $L=10$ density partition intervals, $\Delta H=2$, $B = 50$K burn-in
iterations and a total of 5 million (M) iterations (including the burn-in iterations).
The local proposal was simply $\dnorm(\bfX^t, \mathbf{I})$. 
The average acceptance rate was 0.26 for the local move and was 0.56 for the mixed jump.
Let $\bfX=(X_1,\ldots,X_4)$ and $S=\sum_i X_i$.
We estimated $E(\bfX)$, $E(e^{2S})$, $E(\prod_i X_i)$, $E(\sum_i X_i^5)$, and
$E(\sum_i X_i^6)$, all via domain-based representations. Since the target density of
this example is a product of one-dimensional marginal densities, 
the above expectations can be calculated accurately through one-dimensional numerical integration. 
We compared our estimates from MD sampling with the results from numerical integration by computing
mean squared errors (MSEs). We report the average MSE of 
the estimated log probability masses ($\log \lambda_k$) and the average MSE of the estimated
conditional means ($\bmmu_{\bfX,k}$) over all the local modes at the same layer ($k=1,\ldots,5$), and for other functions
we only report the MSEs of the estimated expectations to save space (Table~\ref{tab:rastrigin}).

\begin{table}[t]
\caption{MSE comparison on the Rastrigin function} \label{tab:rastrigin}
\centering
   \begin{tabular}{cccc|cccc|lccc} 
   \hline
     	  					& MSE	& \multicolumn{2}{c|}{RMSE}&& MSE		& \multicolumn{2}{c|}{RMSE}	&  & MSE		& \multicolumn{2}{c}{RMSE}\\
				   		& MD		& WL			&	MD$_0$		&		& MD			& WL		& MD$_0$	& 		& MD			& WL 	& MD$_0$ \\
  \hline
$\log \lambda_1$ & 1.1e-5	& 1.93	& 2.24	& $\bmmu_{\bfX,1}$	& 2.3e-4	& 0.83 & 2.87	&	$\bfX$					& 1.7e-4 & 1.22 & 3.19 \\
$\log \lambda_2$ & 3.6e-3	& 2.35	& 2.47	& $\bmmu_{\bfX,2}$	& 2.5e-4	& 2.88 &2.96	&	$e^{2S}$	& 0.59 & 3.25 	& 3.09\\
$\log \lambda_3$ & 3.5e-3	& 4.64	& 3.22	& $\bmmu_{\bfX,3}$	& 2.8e-4	& 6.16 &3.34   &	$\prod_i X_i$		& 1.6e-9 & 2.84 & 2.70\\
$\log \lambda_4$ & 3.3e-3	& 8.63	& 3.92	& $\bmmu_{\bfX,4}$	& 2.9e-4	& 12.0 &	4.09 &	$\sum_i X_i^5$		& 6.1e-3 & 2.71 & 2.06\\
$\log \lambda_5$ & 3.2e-3	& 16.8	& 4.66	& $\bmmu_{\bfX,5}$	& 3.3e-4	& 21.1 &	5.06 &	$\sum_i X_i^6$		& 0.11 & 1.95 	& 2.26\\
    \hline
  \end{tabular}
\end{table}

As a comparison, we also applied the WL algorithm (as in the naive two-step approach) 
to this problem with the same parameter setting.
The ratio (RMSE) of the MSE of the WL algorithm over that of the MD sampler for each estimate is given 
in Table~\ref{tab:rastrigin}. The WL algorithm showed larger MSEs than the MD sampler 
for almost all the estimates, especially for those on domains at layers 3, 4 and 5.
For example, the MD sampler was at least 16 times more efficient than the WL algorithm for estimating
$\log \lambda_5$ and $\bmmu_{\bfX,5}$. 
The WL algorithm did not simulate sufficient samples
from these domains, although it visited uniformly different density partition intervals. 
On the contrary, the double-partitioning design facilitated the MD sampler to
explore every domain in a uniform manner, which led to a substantial improvement 
in estimation for these layers. This shows the critical role of domain partitioning in estimating DRs. 
To study the effect of the mixed jump, we re-applied the MD sampler with
$p_{mx}=0$, and calculated the ratio of the resulting MSE (MD$_0$ in Table~\ref{tab:rastrigin}) 
over that of the MD sampler with $p_{mx}=0.1$, the default setting. One sees an increase of two folds or more in 
MSEs without the mixed jump. The convergence of the MD sampler without the mixed jump became slower, reflected
by a five-fold increase in $\gamma_{n}$ after the same number of iterations,
averaging over 100 independent runs. These observations 
demonstrate that the mixed jump served as an efficient global move 
which accelerated convergence of the MD sampler and improved estimation accuracy. 

\section{Learning Bayesian networks}\label{sec:BN}

A Bayesian network (BN) factorizes the joint distribution of $m$ variables
$Z=\{Z_1,\ldots,Z_m\}$ into
\begin{equation}\label{eq:BNdef}
P(Z) = \prod_{i=1}^m P(Z_i \mid \Pi_i^{G}),
\end{equation}
where $\Pi_i^{G} \subset Z $ is the parent set of $Z_i$.
A graph $G$ is constructed to code the structure of a BN by connecting
each variable (node) to its child variables via directed edges. For (\ref{eq:BNdef}) to
be a well-defined joint distribution, the graph $G$ must be a DAG. 
We consider the use of Bayesian networks in causal inference (Spirtes, Glymour and Scheines 1993, Pearl 2000),
which is tightly connected to many areas in statistics, such as structural equations, potential outcomes,
and randomization (Holland 1988, Neyman 1990, Rubin 1978, Robins 1986). Here we follow
Pearl's formulation of causal networks by modeling experimental intervention. If $Z_j$ is a parent of
$Z_i$ in a causal Bayesian network, then experimental interventions that change the value
of $Z_j$ may affect the distribution of $Z_i$, but not conversely. Once all the parents of $Z_i$
are fixed by intervention, the distribution of $Z_i$ will not be affected by interventions on 
any variables in the set $Z\setminus(\Pi_i^{G} \cup \{Z_i\})$. 
In the example causal network of Figure~\ref{fig:examplenet}, 
if we fix $Z_1$ and $Z_3$ by experimental intervention,
then the distribution of $Z_4$ will not be affected by perturbations on $Z_2$, $Z_5$, or $Z_6$.

\begin{figure}[ht]
\centering
\includegraphics[width=1.5in]{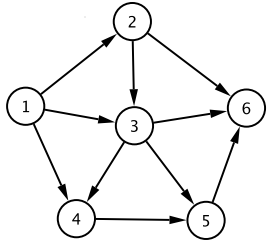}\\ 
\caption{An example Bayesian network of six variables.
\label{fig:examplenet}} 
\end{figure}

\subsection{Posterior distribution}

We focus on the discrete case where each $Z_i$ takes $r_i$ states indexed by $1, \ldots, r_i$
and the parents of $Z_i$ take $q_i = \prod_{Z_j \in \Pi_i^{G}} r_j$ joint states. 
Let $\theta_{ijk}$ be the causal probability for $Z_i=j$ given the 
$k$th joint state of its parent set. 
A causal BN with a given structure $G$ is parameterized
by ${\bmTheta} = \{\theta_{ijk} : \sum_j \theta_{ijk} =1, \theta_{ijk}\geq 0 \}$.

We infer network structure from two types of data jointly, experimental data 
and observational data. For experimental data, a subset of variables are known to be
fixed by intervention. Inferring causality with intervention has been extensively studied in various contexts 
(e.g., Robins 1986, 1987, Pearl 1993). 
We adopt Cooper and Yoo (1999) for calculating
the posterior probability of a network structure given a mix of experimental and observational data. 
Suppose that $N_{ijk}$ is the number of data points for which $Z_i$ is not fixed by intervention
and is found in state $j$ with its parent set in joint state $k$. Then, the collection of counts $\bfN = \{N_{ijk} \}$
is the sufficient statistic for $\bmTheta$ (Ellis and Wong 2008). 
Let $|\Pi_i^{G}|$ be the size of the parent set of $Z_i$.
The prior distribution over network structures is specified as 
$\pi(G) \propto \beta^{\sum_i |\Pi_i^{G}|}$, $\beta \in (0,1)$,
which penalizes graphs with a large number of edges. With a product-Dirichlet prior for
${\bmTheta}$, the posterior distribution $[G \mid \bfN]$ (Cooper and Herskovits 1992) is
\begin{equation}\label{eq:posteriorDAG}
P(G \mid \bfN)  \propto  \prod_{i=1}^m \left\{ \beta^{ |\Pi_i^{G}|} \prod_{k=1}^{q_i}
\left[ \frac{\Gamma(\alpha_{i\cdot k})}{\Gamma(\alpha_{i\cdot k}+N_{i\cdot k})}
\prod_{j=1}^{r_i} \frac{\Gamma(\alpha_{ijk}+N_{ijk})}{\Gamma(\alpha_{ijk})}
\right]
\right\}, 
\end{equation}
where $\alpha_{ijk} = \alpha / (r_i q_i)$ is the pseudo count for the causal probability $\theta_{ijk}$
in the product-Dirichlet prior and $N_{i\cdot k} = \sum_j N_{ijk}$ (similarly for $\alpha_{i\cdot k})$.
The hyperparameters in the prior distributions are chosen as $\beta=0.1$ and $\alpha=1$.

\subsection{MD sampling over DAGs}

The space of DAGs is discrete in nature. We define domains of attraction for $P(G \mid \bfN)$ (\ref{eq:posteriorDAG})
with a move set composed of addition, deletion and reversal of an edge.
Given a DAG $G_a$, we say that another DAG $G_b$ is a neighbor of $G_a$ if $G_b$ can be obtained
via a single move starting from $G_a$, i.e., by adding, deleting or reversing an edge of $G_a$. 
Denote by $ngb(G_a)$ all the neighbors of $G_a$ and let $\overline{ngb}(G_a) = ngb(G_a)\cup \{G_a\}$. 
A DAG $G^*$ is defined as a local mode of 
a probability density (mass) function $p(G)$ if $p(G^*)>p(G')$ for every $G' \in ngb(G^*)$. Let $G^0$ be a DAG
and define recursively 
\begin{equation}\label{eq:sna}
G^{t+1}=\underset{G\in \overline{ngb}(G^t)}{\amax} \,p(G), \mbox{ for $t=0,1,\ldots$.}
\end{equation}
That is, we recursively find the single move that leads to the greatest increase in
$p$ until a local mode is reached, which can be viewed as a discrete counterpart of
the gradient descent algorithm. If there are more than one maximum in (\ref{eq:sna})
with an identical function value, we take the first maximum according to a fixed ordering of the neighbors.
We call this recursion the steepest neighbor ascent (SNA). Based on SNA, we define
the domain partition index $I(G)$ and the attraction domains of $p$ in the same manner as
for a differentiable density (Equation \ref{eq:defI} and Definition \ref{def:Domain}). 

The target distribution in this application is the posterior distribution $P(G \mid \bfN)$ and
we define working density $P(G \mid \bfN; \bfW)$ similarly as in (\ref{eq:weightds}). 
To implement the MD sampler for DAGs, we employ the move set as the local
proposal and develop the following mixed jump. 
For a DAG $G$, we define an edge variable $E_{ij}^G$ for every pair of nodes $Z_i$ and $Z_j$ $(i<j)$ 
such that $E^G_{ij} = 1$ if $Z_i$ is a parent of $Z_j$, $E^G_{ij}=-1$ if $Z_j$ is a parent of $Z_i$,
and $E^G_{ij}=0$ otherwise.
Given a DAG $\nu$, let $\mathbf{C}(G; \nu)=(C_a(G;\nu),C_d(G;\nu),C_r(G;\nu))$ be a map of $G$,
where
\begin{equation*}
\begin{array}{lll}
C_a(G;\nu) & = & \sum_{i<j} \bm{1}(E^G_{ij}\ne 0, E^{\nu}_{ij}=0), \\
C_d(G;\nu) & =  &\sum_{i<j} \bm{1}(E^G_{ij}= 0, E^{\nu}_{ij}\ne 0), \\
C_r(G;\nu) & = &\sum_{i<j} \bm{1}(E^G_{ij}\cdot E^{\nu}_{ij}=-1).
\end{array}
\end{equation*}
In words, $\mathbf{C}(G;\nu)$ gives the numbers of additions, deletions and reversals needed
to obtain $G$ from $\nu$. Let $T=m(m-1)/2$ be the total number of node pairs and
$|E_{\nu}|$ be the number of edges in $\nu$. Then, the number of common edges between $G$ and $\nu$ is
$|E_{\nu}|-[C_r(G;\nu)+C_d(G;\nu)]$ and the number of node pairs with no edge in either DAG
is $T-|E_{\nu}|-C_a(G;\nu)$.
Given a set of local modes of $P(G \mid \bfN)$, $\{\nu_k\}_{k=1}^M$,
let $\bfg_k(G) = \mathbf{C}(G; \nu_k)$ in the update of $\bfV^t_k$ (\ref{eq:updateV}) of Algorithm~\ref{alg:main},
where $\bfV^t_k=(v^t_{k,a},v^t_{k,d},v^t_{k,r})$ is a vector. As $n\to \infty$, 
$\bfV^n_k \toas E[\mathbf{C}(G;\nu_k)\mid G \in D_k]$, 
where the expectation is taken with respect to the limiting working density $P(G \mid \bfN; \bfW^*)$ (\ref{eq:Vlimit}).
In the mixed jump, after a local mode $\nu_k$ is randomly chosen,
we sequentially modify the edge variables of $\nu_k$ to propose a new DAG $Y$.
The proposal is designed according to $\bfV^t_k$,
the current estimate of the expected numbers of additions, deletions and reversals
of DAGs in the domain $D_k$ relative to the mode $\nu_k$. 
Let $|E_k|$ be the number of edges in $\nu_k$.
If $E^{\nu_k}_{ij}\ne 0$, we propose to reverse, delete, or retain the edge $E^{\nu_k}_{ij}$
(i.e., $E^{Y}_{ij}= -E^{\nu_k}_{ij}, 0, \mbox{or } E^{\nu_k}_{ij}$) with probabilities proportional to the vector
$(v^t_{k,r},v^t_{k,d},|E_k|-(v^t_{k,r}+v^t_{k,d}))+b$, where $b>0$ is a small
prior count added to each category. Analogously, if $E^{\nu_k}_{ij}= 0$ we propose 
$E^{Y}_{ij}= 0,1, \mbox{or } -1$ with probabilities proportional to $(T-|E_k|-v^t_{k,a},v^t_{k,a}/2,v^t_{k,a}/2)+b$.
Lastly, to ensure a proposed graph is acyclic, a check for cycles is performed
when we propose to add or reverse an edge in either the local proposal or the mixed jump. If the resulting
graph is cyclic, we suppress the probability for the corresponding move.

Following the common practice in structural learning of discrete BNs, we set an upper bound
for the number of parents (indegree) of a node. 
In all the following examples and applications, this upper  bound is chosen to be four. 
We are interested in the posterior expected adjacency matrix $\bfA=(a_{ij})_{m\times m}$ and its 
domain-based representation,
where $a_{ij}$ $(1\leq i,j \leq m)$ is the posterior probability for a directed edge from $Z_i$ to $Z_j$. 
For each identified local mode $\nu_k$,
we estimate the probability mass $\lambda_k$ of its attraction domain $D_k$
and the conditional expected adjacency matrix $\bfA_k$ on the domain. 
Then, $\bfA$ is estimated by $\hat{\bfA} = \sum_k \hat{\lambda}_k \hat{\bfA}_k$.

\subsection{Simulation}

We simulated data from two BNs, each of
six binary variables ($m=6,r_i=2, \forall i$). This is the maximum number of nodes for which we can enumerate
all DAGs, numbering about four million, to obtain true posterior distributions and domain-based representations 
as the ground truth for testing a computational method. The first network has a chain structure
in which $Z_i$ is the only parent of $Z_{i+1}$ for $i=1,\ldots,5$ and  $Z_1$ has no parent.
The second network has a more complex structure shown in Figure~\ref{fig:examplenet}.
We simulated 50 datasets independently from each network. In each dataset, 20\% of the
data points were generated with interventions. Please see supplementary document for 
data simulation details.

The MD sampler was applied to the 100 datasets with $L=15$,
$\Delta H=10$, $p_{mx}=0.1$, $K^*=100$ and a total of 5M iterations with the first 50K
as burn-in iterations. To verify its performance, we compared identified local modes,
estimated probability masses $\{ \log \hat{\lambda}_k\}$, conditional expected adjacency
matrices $\{\hat{\bfA}_k\}$, and expected adjacency matrix $\hat{\bfA}$
to their respective true values obtained via enumerating all DAGs.
Our enumeration confirms that the posterior distributions indeed have multiple local modes.
The chain and the graph (Figure~\ref{fig:examplenet}) networks have on average 3.57 
and 7.06 modes over the simulated datasets, respectively, and the maximum number of modes
is 29 for the chain and 34 for the graph. As reported in Table~\ref{tab:BNsim}, the MD sampler
did not miss a single local mode for any dataset, which demonstrates its global search ability. 
Recall that all recorded modes are detected in the burn-in algorithm.
In fact, all modes, including the global mode, were identified within 10K iterations for every dataset. 
This observation confirms the notion that the burn-in algorithm alone may serve as a powerful
optimization algorithm (Remark~\ref{rmk:burninopt}).
We calculated the MSE of the vector $(\log \hat{\lambda}_1,\ldots,\log \hat{\lambda}_K)$,
where $K$ is the number of local modes, and the average MSE of $\hat{\bfA}_1,\ldots,\hat{\bfA}_K$.
When calculating the MSE of the log probability vector, we ignored those tiny
domains with a probability mass $<10^{-4}$. 
These estimates are seen to be very accurate as reported in Table~\ref{tab:BNsim}.

\begin{table}[t] 
\caption{Comparison on simulated data from two BNs} \label{tab:BNsim}
\begin{center}
   \begin{tabular}{ccccc} 
   \hline
 									&	MD		& MD$_0$		& WL & $K^*=10$ \\
									& MSE		& \multicolumn{3}{c}{RMSE} \\
	\hline
\# of missed modes			&  0		& 0		& 0     & 0.51\\ 
$\log\hat{\lambda}_k$ 		& 0.028  & 1.48		& 5.06 & 0.48\\
$\hat{\bfA}_k$						& 1.3e-4 & 1.65		& 11.0  & 0.99\\
$\hat{\bfA}$							& 1.3e-4  & 1.76	& 1.67 & 1.03\\
\hline
\# of missed modes			& 0 		& 0		& 0.12  & 2.06\\
$\log\hat{\lambda}_k$			& 0.029	& 401		& 1368 & 0.84\\
$\hat{\bfA}_k$						& 1.7e-4 & 2.96		& 13.6  & 1.01\\
$\hat{\bfA}$							& 1.5e-4 & 7.17		& 5.09  & 0.98\\
\hline
  \end{tabular}
\end{center}
The top and bottom panels report the results for the chain and the graph networks, respectively. 
 For each estimate, reported are the MSE of the MD sampler and the RMSEs (ratios) of the other methods
 relative to the MD sampler.
\end{table}

We also applied the MD sampler with $p_{mx}=0$ (MD$_0$) and the WL algorithm
with the same parameter setting to these datasets (Table~\ref{tab:BNsim}). 
The degraded performance of MD$_0$ demonstrates 
the effectiveness of the mixed jump for sampling DAGs. The WL algorithm
missed 0.12 modes on average for the second network, and its estimation of the DR $\{(\hat{\bfA}_k, \hat{\lambda}_k)\}$ 
was much less accurate compared to that of the MD sampler.
The average MSE of $\hat{\bfA}_1,\ldots,\hat{\bfA}_K$ and the MSE of $(\log \hat{\lambda}_k)_{1:K}$ 
were more than 10 and 1,000 times greater than those of the MD sampler, respectively.
The huge MSE of the $(\log\hat{\lambda}_k)$ constructed by the WL algorithm was often
due to severe underestimation of the probability masses of domains sampled insufficiently.
This result implies that without domain partitioning, the WL algorithm 
is unable to estimate domain-based representations for a BN of
a moderately complicated structure.
Since the number of local modes often increases very fast with the complexity of a problem,
we re-applied the MD sampler with $K^*=10$ to investigate the effect of keeping only a subset of local modes.
Obviously, the algorithm missed a few local modes when the total number of modes exceeded $K^*$.
But in terms of estimating $\bfA$ and $\bfA_k$, the performance of the MD sampler with $K^*=10$
was very comparable to its performance when all the local modes were kept (Table~\ref{tab:BNsim}).
The probability mass outside the domains of recorded modes, $\lambda_0 = 1-\sum_{k=1}^{K^*} \lambda_k$,
is less than $0.007$ averaging over the 15 datasets where the posterior distributions have
more than 10 local modes. This confirms that the MD sampler indeed captured
major modes in the burn-in period.

\section{Protein-signaling networks}\label{sec:signaling}

\subsection{Background and data}

The ability of cells to properly respond to environment 
is the basis of development, tissue repair, and immunity.
Such response is established via information flow along signaling pathways
mediated by a series of signaling proteins.
Cross-talks and interplay between pathways
reflect the network nature of the interaction among these signaling molecules. Construction
of signaling networks is an important step towards a global understanding of normal cellular responses
to various environmental stimuli and more effective treatment of complex diseases caused by
malfunction of components in a pathway.
Causal Bayesian networks may be used for modeling signaling networks
as the relation among pathway components has a natural causal interpretation. That is,
the activation or inhibition of a set of upstream molecules in a network
causes the state change of downstream molecules. 
An edge from molecule A to molecule B in a signaling network implies that a change in the state of A
causes a change in the state of B via a direct biochemical reaction.
Here, a state change refers to a chemical, physical or locational modification of a molecule.
However, as there may exist mutual regulation between two signaling molecules, the use
of DAGs for modeling signaling networks is only a first-step approximation.

In this study, we construct protein-signaling networks
from flow cytometry data. Polychromatic flow cytometry is a high-throughput 
technique for probing simultaneously the (phosphorylation) states of multiple proteins in a single cell. Since
measurements are collected on a cell-by-cell basis, huge amounts of data can be produced in one
experiment. Sachs et al. (2005) made flow cytometry measurements of 11 proteins and phospholipids 
in the signaling network of human primary naive CD4$^+$ T cells under nine different experimental
perturbations that either activate or inhibit a particular molecule or activate the entire pathway.
Note that a perturbation that activates or inhibits a particular molecule is essentially an
intervention on the molecule, so that causal structures of the underlying network may be inferred. 
Under each perturbation, 600 cells were collected with 11 measurements for each.
The measurements in the data were discretized into three levels, high, medium, and low by Sachs et al.
In summary, this dataset contains 5,400 data points for 11 ternary variables.
Since naive T cells are essential for the immune system to continuously respond to unfamiliar pathogens,
extensive studies have been conducted to establish the signaling pathways.
An annotated signaling network
among the 11 molecules, provided by Sachs et al., 
is depicted as a causal Bayesian network in Figure~\ref{fig:cytonet}. This network
contains 18 edges that are well-established in the literature and two edges (PKC $\to$ PKA and Erk $\to$ Akt) reported
from recent experiments independent of the flow cytometry data.
 
\begin{figure}[th]
\centering
\includegraphics[width=3in]{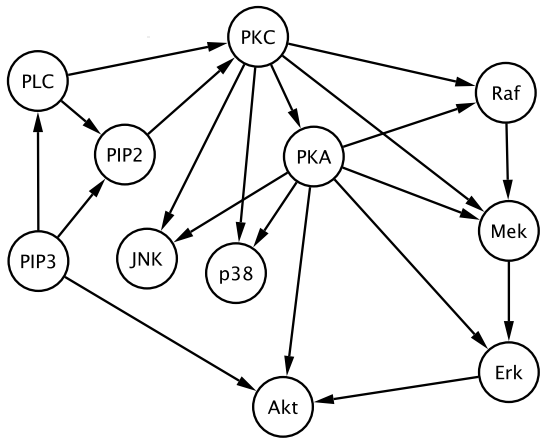}\\ 
\caption{An annotated protein-signaling network in naive CD4$^+$ T cells.
\label{fig:cytonet}} 
\end{figure}

\subsection{Predicted networks}

The MD sampler was applied to this dataset with $L=20$, $\Delta H=10$, $p_{mx}=0.1$, and
$K^*=10$. The total number of iterations was 5M, of which the first 50K were used for burn-in.
We estimated the posterior expected adjacency matrix $\bfA$ and its domain-based representation.
Three predicted networks were constructed by thresholding posterior edge probabilities at $c=0.5, 0.7, 0.9$,
i.e., an edge from $Z_i$ to $Z_j$ was predicted if the edge probability $\hat{a}_{ij}\geq c$.
For simplicity we call such a predicted network a mean network (with a threshold $c$).
Table~\ref{tab:cyt} (top panel) reports the number of true positive edges (TP) that are both 
predicted by the MD sampler and annotated in Figure~\ref{fig:cytonet}
and the number of false positive edges (FP) that are predicted but not annotated, 
together with the (unnormalized) log posterior probability of the identified global maximum DAG. 
To compare the results, we re-applied the MD sampler with the 
same parameters except that $p_{mx}=0$ (MD$_0$ in Table~\ref{tab:cyt}) and applied
the WL algorithm with the same parameters as used in the MD sampler to the same data.
The average result over 20 independent runs of each method is summarized in Table~\ref{tab:cyt}.
In terms of finding the global mode, the MD sampler was much more effective and
robust than the other two algorithms, reflected by a much higher average log probability and
a much smaller standard deviation across multiple runs.
The MD sampler with or without the mixed jump showed comparable results in predicting network structures,
and both predicted more true positives and fewer false positives than 
the WL algorithm did for all the three thresholds.
We noticed that the mean networks constructed with different thresholds ($c=0.5,0.7,0.9$) were almost identical.
This was due to the fact that the posterior edge probabilities were close to either 1 or 0
because of the large data size. 
The network constructed from the same data by
the order-graph sampler, reported in Figure 11 of Ellis and Wong (2008), has
9 true positive and 11 false positive edges, which misses much more true edges
and includes more false positives than the networks predicted by the MD sampler.
These results demonstrate that the MD sampler is very powerful in learning underlying
network structures from experimental data compared to other advanced Monte Carlo techniques.

\begin{table}[t] 
\caption{Results on the flow cytometry data} \label{tab:cyt}
\begin{center}
   \begin{tabular}{cccc} 
   \hline
 			&MD		& MD$_0$		& WL  \\
	\hline
Global max (SD)	&$-31757.9$(2.7)	&$-31787.3$(82)	& $-31937.8$(199) \\ 
TP/FP($c=0.5$) 		&15.50/10.35		&15.55/10.35		& 13.40/12.35 \\
TP/FP($c=0.7$) 		&15.50/10.35		&15.55/10.35		& 13.35/12.25 \\
TP/FP($c=0.9$) 		&15.50/10.35		&15.55/10.35		& 13.35/11.95 \\
\hline
TP/FP($c=0.5$) 		&15.2/10.3		&14.6/10.6		& 11.6/13.6 \\
TP/FP($c=0.7$) 		&15.2/10.2		&14.6/10.6		& 11.6/13.2 \\
TP/FP($c=0.9$) 		&15.2/9.9		&14.6/10.2		& 11.6/13.1 \\
Log pred (mean)	&0				& $-1.4$ 		&$-62.1$ \\ 
Log pred (DR)		&17.5 		&16.2  		    	&$-51.7$ \\
\hline
  \end{tabular}
\end{center} 
The top and bottom panels report the average results over 20 independent runs 
on the full dataset and the average results over ten test datasets in cross validation, respectively. 
Predictive probabilities (Log pred) are reported as log ratios over the predictive probability 
given the mean network of the MD sampler.
\end{table}

Next, we focus on the estimation of the DR, 
$\{(\hat{\bfA}_k,\hat{\lambda}_k): k=0,\ldots,K^*\}$, and its scientific implications.
A network $\hat{G}_k$ can be
constructed for an attraction domain by thresholding $\hat{\bfA}_k$,
the conditional expected adjacency matrix on the domain $D_k$, for $k=1,\ldots,K^*$.
To distinguish it from the mean network, we call $\hat{G}_k$ a local network.
We take the result of a representative run of the MD sampler ($p_{mx}=0.1$)
to demonstrate local networks with the threshold $c=0.9$.
The parent sets of eight nodes are identical across the $K^*=10$ local networks.
We report in Table~\ref{tab:cytDR} the parents of the other three nodes, PLC, PIP3, and Erk,
which are distinct among the local networks, together with the probability masses ($\log\hat{\lambda}_k$) 
of the 10 domains and the
probabilities of the local modes [$\log P(\nu_k \mid \bfN)$]. 
The local networks may predict meaningful alternative edges not included in the mean network,
as illustrated by the result on a particular pathway, Raf $\to$ Mek $\to$ Erk (Figure~\ref{fig:cytonet}). 
This expected pathway was predicted by all the local networks
and the mean network. However, some local networks also contained a direct link from
Raf to Erk (Table~\ref{tab:cytDR}). 
As Mek was inhibited in one of the experimental conditions, this finding suggests that the cells may have
another pathway that passes the signal from Raf to Erk via some indirect regulation or via molecules
not included in this analysis, when Mek is not functioning properly. 
Such compensational mechanisms exist widely in many biological networks. 
Indeed, Raf has been reported to enhance the kinase activity of PKC$\theta$,
an isoform of PKC, although PKC$\theta$ is unlikely a direct phosphorylation target of Raf (Hindley and Kolch 2007).
As indicated by Figure~\ref{fig:cytonet}, Erk is a downstream node of PKC
and thus may be regulated indirectly by Raf via the enhanced kinase activity of PKC$\theta$.
Such novel hypotheses could not be proposed if we did not construct the DR for the posterior distribution.
Clearly, the DR
of network structures not only gives a detailed landscape of various local domains
but also provides new insights into the underlying scientific problem.

From Table~\ref{tab:cytDR} we find that the probability mass is dominated by the domain of the identified global
mode with a log probability of $-31757$. Consistent with the summary in Table~\ref{tab:cyt},
the MD sampler almost always reached this global mode for different runs.
On the contrary, the highest mode detected by the WL algorithm with an average log probability of $-31938$
is even much lower than the lowest mode in Table~\ref{tab:cytDR}. In other words,
the WL algorithm was inevitably trapped to some local modes with negligible probability masses.
This again demonstrates the advantage of the MD sampler, particularly the burn-in algorithm, 
in finding global modes.
Even when the probability mass of the global mode is dominant and
other domains occupy only a small fraction of the sample space, without the domain-partitioning
design the WL algorithm may be trapped to a local mode of a tiny probability mass and produce severely 
biased estimates. 

\begin{table}[t]
\caption{Local networks constructed from domain-based representation} \label{tab:cytDR}
\centering
   \begin{tabular}{cc|lll} 
   \hline
   									& 							&	\multicolumn{3}{c}{Parents of} \\
 	$\log\hat{\lambda}_k$		&$\log P(\nu_k \mid \bfN)$		& PLC	& PIP3 & Erk  \\
	\hline
$-$0.00274 &$-$31756.81 & PKC p38 JNK   & PLC  PKC     & Mek PKA PKC  \\
$-$6.39777 &$-$31762.75 & PKC p38 JNK   & PLC  PKC     &  Mek PKA Raf  \\
 $-$7.12166 &$-$31764.16 & Mek  PKA      & PLC  PKC     & Mek PKA PKC  \\
$-$8.34521 &$-$31764.82 & PKC p38 JNK   & PLC JNK     & Mek PKA PKC  \\
$-$10.6461 & $-$31766.81 & Mek  Akt p38   & PLC  PKC     & Mek PKA PKC   \\
$-$13.7193 &$-$31770.11 & Mek  Akt  PKA   & PLC  PKC     &  Mek PKA Raf  \\
$-$14.0777 &$-$31770.77 & PKC p38 JNK   & PLC JNK     &  Mek PKA Raf  \\
$-$15.0459 &$-$31772.18 & Mek  Akt  PKA   & PLC JNK     & Mek PKA PKC  \\
$-$16.8860 &$-$31772.76 & Mek  Akt p38   & PLC  PKC     &   Mek PKA Raf \\
$-$18.5716 &$-$31774.82 & Mek  Akt p38   & PLC JNK     & Mek PKA PKC  \\
\hline
  \end{tabular}
\end{table}

\subsection{Cross validation}

To check the statistical variability and the predictive power of our method,
we conducted ten-fold cross validation on this dataset. We partitioned randomly the 5,400 data
points into ten subsets of equal sizes. We used nine subsets as training data to learn a mean network 
and calculated the predictive probability of data points in the other subset (test data) given the
learned mean network.
This procedure was repeated 10 times to test on every subset. 
We verified the average accuracy of the mean networks constructed from the 10 training datasets with different thresholds. 
The performance of the MD sampler on the training datasets was comparable
to its performance on the full dataset, which implies its robustness to random sampling of input data.
The improvement in accuracy (TP/FP) of the MD sampler over the other two methods became
more significant, especially compared to the WL algorithm (Table~\ref{tab:cyt}, bottom panel).
The average predictive probability of the test datasets 
given the mean networks constructed 
from the training datasets by each method with $c=0.9$
is reported in Table~\ref{tab:cyt} [Log pred (mean)], 
from which we see that the predictive power of the networks constructed by the two
MD samplers was much higher than that of the WL algorithm ($>60$ in log probability ratio). 
In addition, we utilized the estimated domain-based representations
to calculate the predictive probability of a test data point $\bfy$ by
\begin{equation}\label{eq:DRgraph}
P(\bfy \mid \{\hat{G}_k,\hat{\lambda}_k\}) \defi \sum_{k=0}^{K^*} \hat{\lambda}_k P(\bfy \mid \hat{G}_k),
\end{equation}
where $P(\bfy \mid \hat{G}_k)$ is the marginal likelihood of $\bfy$ given the local network $\hat{G}_k$.
This can be regarded as a domain-based approximation to the posterior predictive distribution
\begin{equation*}
P(\bfy \mid \yobs) = \sum_G P(G \mid \yobs) P(\bfy \mid G),
\end{equation*}
i.e., we use estimated probability masses $\hat{\lambda}_k$ and conditional mean networks $\hat{G}_k$ 
to approximate the posterior predictive probability. The advantage is that there is no need to store
a large posterior sample of networks but only an estimated DR. Since Equation \eqref{eq:DRgraph} captures
the variability among different domains, it is expected to outperform the mean network in prediction.
In fact, for each method the predictive probability calculated by (\ref{eq:DRgraph}) [Table~\ref{tab:cyt}, Log pred (DR)]
was indeed significantly greater than the predictive probability calculated given its mean network,
especially for the two MD samplers.

In real applications, we are interested in predicting results for a new experimental condition
given observed data from other conditions. Thus, we also performed a nine-fold cross validation
where a training dataset was composed of cells from eight experimental conditions and a test dataset
only included cells in the other one condition. We applied the MD sampler to construct mean networks 
and DRs $\{\hat{G}_k,\hat{\lambda}_k\}_{k=1}^{10}$ from the training datasets.
The mean networks with $c=0.9$ included, on average, 13.2 true edges with 10.8 false edges,
which was slightly worse than the result from the ten-fold cross validation.
The degraded performance is expected as removing
all cells from one experimental perturbation will increase the uncertainty in determining the directionality of
the network. The domain-based prediction (\ref{eq:DRgraph}) 
was compared against the annotated network $G^*$ given in
Figure~\ref{fig:cytonet}, which presumably has the highest predictive power, by evaluating the log likelihood ratio (LLR)
$\log R = \log [P(\bfy \mid \{\hat{G}_k,\hat{\lambda}_k\})/P(\bfy \mid G^*)]$, where $\bfy$ is a test data point.
The average $\log R$ over all test data points was $-0.062$, and thus the predictive
probability for a new observation given the constructed DR is expected to be higher than $94\% (=e^{-0.062})$
of its likelihood given the annotated graph. This demonstrates the high predictive power of the domain-based
prediction constructed by our method. As expected, the average LLR of the mean networks over $G^*$ was 26\%
lower than the average of $\log R$.

\section{Discussion}

The central idea of this article is to construct domain-based representations with
the MD sampler. Related works have been seen in the physics literature under the name
of superposition approximation. Please see Wales and Bogdan (2006) for a recent review.
Given a Boltzmann distribution $p_B(\bfx;\tau) \propto \exp[-H(\bfx)/\tau]$, a superposition approach
identifies the local minima of $H(\bfx)$, i.e., the local modes of $p_B(\bfx;\tau)$, and approximates $H(\bfx)$
on the attraction domain of a local minimum by quadratic or high-order functions.
The approximation is often proposed based on expert knowledge about the physical model under study.
Expectations with respect to $p_B(\bfx;\tau)$ are then estimated by summing over approximations
from identified domains. The accuracy of this approach largely depends on the employed 
approximation to $H(\bfx)$ on a domain and thus may not work well for an arbitrary distribution. 
The MD sampler differs in that domain-based
representations are constructed by Monte Carlo sampling which is able to provide accurate
estimation with large-size samples; no expert knowledge about the target distribution is needed.
In addition, our method also contains a coherent component for
finding local modes, while the superposition approximation works more like a two-step approach.

From a computational perspective, the MD sampler integrates Monte Carlo and deterministic
optimization. A few other methods also have the two ingredients, such as
Monte Carlo optimization (Li and Scheraga 1987), the basin hopping algorithm (Wales and Doye 1997), 
and conjugate gradient Monte Carlo (CGMC) (Liu, Liang, and Wong 2000). 
In Monte Carlo optimization and the basin hopping algorithm, 
the target distribution $p(\bfx)$ is modified to $\tilde{p}(\bfx) = p(\bmnu_k)$ for all $\bfx \in D_k$,
where $D_k$ is the attraction domain of the mode $\bmnu_k$. Then a Metropolis-type MCMC is used
to sample from $\tilde{p}$, in which a local optimization algorithm is employed at each iteration to
find $\tilde{p}(\bfX^t)$ for the current state $\bfX^t$. 
These methods have been applied to identification of minimum-energy
structures of proteins and other molecules. However, its application to other fields is
limited as the modified density $\tilde{p}(\bfx)$ may be improper when the sample space is unbounded.
In CGMC, a population of samples is evolved and 
a line sampling step (Gilks, Roberts, and George 1994) is performed on a sample 
along a direction pointing to a local mode found by local optimization initiated at another sample.
In this way promising proposal may be constructed 
by borrowing local mode information from other samples. A possible future work on the MD sampler
is to utilize a population of samples. Because local modes are recorded, similar proposals as the line
sampling can be developed for the MD sampler to further enhance sampling effectiveness.
Another future direction is to construct disconnectivity graphs (Becker and Karplus 1997) or 
trees of sublevel sets (Zhou and Wong 2008) from samples generated by the MD sampler.
Since samples have been partitioned into domains of attraction, 
we only need to determine the barrier between a pair of domains, defined by
$\max_{s \in \mathcal{S}}\min_{\bfx\in s} p(\bfx)$, where $\mathcal{S}$ is the collection of 
all the paths between the two domains.
A few candidate approaches towards this direction are under current investigation (Zhou 2011).

\section*{Appendix: Theoretical analysis}\label{sec:app}

In this Appendix we establish the convergence and ergodicity properties of the MD sampler.
Our analysis is conducted for a doubly adaptive MCMC, i.e.,
both the target distribution and the proposal may change along the iteration,
which includes the MD sampler as a special case.
Furthermore, the MWL design (Routine \ref{MWLdesign}) is employed to adjust $\gamma_t$.

Assume that the sample space $\calX$ is equipped with a countably generated $\sigma$-field, $\calB(\calX)$.
Let $\{\calX_i\}_{i=1}^{\kappa}$ be a partition of $\calX$, where each $\calX_i$ is nonempty, and
$\calB(\calX_i)=\{A \in \calB(\calX): A \subseteq \calX_i\}$ for $i=1,\ldots,\kappa$.
Let $\bmomega=(\omega_i)_{1:\kappa} \in \Omega$ and $\bmphi \in \Phi$ be two 
vectors of real parameters. Denote the product parameter space by
$\Theta = \Omega \times \Phi$ and write $\bmtheta = (\bmomega,\bmphi) \in \Theta$.
For $\bmomega \in \Omega$, define working density 
\begin{equation}\label{eq:pomega}
p_{\bmomega}(\bfx) \propto \sum_{i=1}^{\kappa} e^{-\omega_i} p(\bfx) \I(\bfx \in \calX_i).
\end{equation}
Let $q(\bfx,\cdot)$ and $t_{\bmphi}(\bfx,\cdot)$, $\bmphi \in \Phi$, be two transition kernels on $(\calX, \calB(\calX))$.
Hereafter, the same notation will be used for a kernel and its density
with respect to the Lebesgue measure on $\calX$, e.g., $q(\bfx,d\bfy)\equiv q(\bfx,\bfy)d\bfy$.
For $j = 0,1$, define $Q_{j,\bmphi}(\bfx,\cdot) = (1-j) q(\bfx,\cdot) + j\, t_{\bmphi}(\bfx,\cdot)$.
Let $K_{j,\bmtheta}$ be the MH transition kernel with $p_{\bmomega}$ as the target distribution
and $Q_{j,\bmphi}$ as the proposal, i.e.,
\begin{equation*}
K_{j,\bmtheta}(\bfx,d\bfy) = S_{j,\bmtheta}(\bfx,d\bfy) + \I (\bfx \in d\bfy)\left[1-\int_{\calX} S_{j,\bmtheta}(\bfx,d\mathbf{z}) \right],
\end{equation*}
where $S_{j,\bmtheta}(\bfx,d\bfy)=Q_{j,\bmphi}(\bfx,d\bfy)\min[1,p_{\bmomega}(\bfy)Q_{j,\bmphi}(\bfy,\bfx)/p_{\bmomega}(\bfx)Q_{j,\bmphi}(\bfx,\bfy)]$,
representing an accepted move.
As $Q_{0,\bmphi}=q$, $K_{0,\bmtheta}$ and $S_{0,\bmtheta}$ do not depend on $\bmphi$
and thus reduce to $K_{0,\bmomega}$ and $S_{0,\bmomega}$, respectively. 
Furthermore, if we let $\bmomega=\mathbf{0}$ then $p_{\bmomega}(\bfx)=p(\bfx)$, in which case
we simply use $K_0$ and $S_0$. 
Given $\alpha \in [0,1)$, define a mixture proposal $Q_{\bmphi} = (1-\alpha) q + \alpha t_{\bmphi}$,
its accepted move $S_{\bmtheta} = (1-\alpha) S_{0,\bmomega} + \alpha S_{1,\bmtheta}$,
and the corresponding MH kernel $K_{\bmtheta} = (1-\alpha) K_{0,\bmomega} + \alpha K_{1,\bmtheta}$.
Table~\ref{tab:notation} summarizes the notations, from left to right, 
for target distributions, proposals, MH kernels, and accepted moves 
for different scenarios involved in this analysis.

\begin{table}[h]
\caption{Summary of notations\label{tab:notation}}
\centering
\begin{tabular}{rcccc} 
\hline
mixture : & $p_{\bmomega}$ & $Q_{\bmphi}$ & $K_{\bmtheta}$ &  $S_{\bmtheta}$ \\ 	
$j=0$ :   & $p_{\bmomega}$  &  $q$ & $K_{0,\bmomega}$ &  $S_{0,\bmomega}$ \\
$j=0,\bmomega=\mathbf{0}$ : & $p$ &  $q$  & $K_{0}$ &  $S_{0}$ \\
$j=1$ : & $p_{\bmomega}$  &  $t_{\bmphi}$ & $K_{1, \bmtheta}$ &  $S_{1, \bmtheta}$ \\
\hline
\end{tabular}
\end{table}

Denote by $Z(\bmomega)$ the normalizing constant of \eqref{eq:pomega}.
Then $Z(\bmomega)=\sum_{i=1}^{\kappa} Z_i(\omega_i)$, where
$Z_i(\omega_i)=e^{-\omega_i}\int_{\calX_i} p(\bfx)d\bfx$. 
Let $U$ be a nonempty subset of $\{1,\ldots,\kappa\}$ and $\calX_U = \bigcup_{i\in U}\calX_i$.
Given a map $\bfg: \calX \to \Phi$, let $\bmmu_{\bfg,U}(\bmomega)=E[\bfg(\bfX)\mid \bfX\in \calX_U]$ 
with respect to $p_{\bmomega}$.
Define a map $\bfH: \Theta \times \calX \to \Theta$ by
\begin{equation*} 
\bfH(\bmtheta,\bfx) = [(\I(\bfx \in \calX_i)-1/\kappa)_{1:\kappa}, (\bfg(\bfx)-\bmphi)\I(\bfx\in\calX_U)]
\end{equation*}
and the mean field $\bfF(\bmtheta) = \int_{\calX} \bfH(\bmtheta,\bfx)p_{\bmomega}(\bfx) d\bfx$, i.e.,
\begin{equation*} 
\bfF(\bmtheta)= \left[\left(\frac{Z_i(\omega_i)}{Z(\bmomega)}-\frac{1}{\kappa}\right)_{1:\kappa}, 
\frac{\sum_{u\in U}Z_u(\omega_u)}{Z(\bmomega)}(\bmmu_{\bfg,U}(\bmomega) - \bmphi) \right].
\end{equation*}
Consider the equation $\bfF(\bmtheta)=\mathbf{0}$.
Let $\bmomega^*= [\log Z_i(0)]_{1:\kappa}$ and $\bmphi^*=\bmmu_{\bfg,U}(\bmomega^*)$. 
As $\bfF(\bmtheta)$ is invariant to translation of $\bmomega$ by a scalar:
$\bmomega \to (\bmomega+\beta) \defi (\omega_i+\beta)_{1:\kappa}$, $\beta \in \R$,
the solution set to this equation is $\{\bmtheta^*(\beta)\defi (\bmomega^*+\beta,\bmphi^*)\}\cap \Theta \defi \Theta^*$.
Set $\gamma_1=1$ and choose an arbitrary point $(\tilde{\bfx},\tilde{\bmtheta}) \in \calX \times \Theta$ to initialize 
$(\bfX^1,\bmtheta^1)$.
A doubly adaptive MCMC is employed to find a solution $\bmtheta^*(\beta) \in \Theta^*$ 
and to estimate $\mu_h(\bmomega^*)=\int_{\calX} h(\bfx)p_{\bmomega^*}(\bfx)d\bfx$ for a function $h:\calX \to \R$.
\begin{algorithm}[Doubly adaptive MCMC]\label{alg:damcmc}
Choose a fixed $\alpha\in [0,1)$. 
For $t=1,\ldots,n$:
\begin{enumerate}
\item{If $\bmtheta^t \notin \Theta$, set $\bfX^{t+1}=\tilde{\bfx}$ and $\bmtheta^{t+1}=\tilde{\bmtheta}$;
otherwise draw $\bfX^{t+1}\sim K_{\bmtheta^t}(\bfX^t,\cdot)$ 
and set $\bmtheta^{t+1}= \bmtheta^t + \gamma_t \bfH(\bmtheta^t,\bfX^{t+1})$.}
\item{Determine $\gamma_{t+1}$ by the MWL design in Routine~\ref{MWLdesign} with $\{\calX_i\}$ in place of
$\{D_{kj}\}$.}
\end{enumerate}
\end{algorithm}

Denote the $L_2$ norm by $|\cdot|$ and let $d(\bfx,A)\defi \inf_{\bfy\in A}|\bfx-\bfy|$, where $\bfx,\bfy$ are vectors 
and $A$ is a set.
Our goal is to establish that $d(\bmtheta^n, \Theta^*) \to 0$ 
almost surely (with respect to the probability measure of the 
process $\{\bfX^t,\bmtheta^t\}$) and that $\{\bfX^t\}$ is ergodic.
Clearly, translation of $\bmomega^t$ by a scaler does not change the working density $p_{\bmomega^t}$ or
affect the convergence of $\bmtheta^t$ to $\Theta^*$.
Thus, the theory for Algorithm~\ref{alg:damcmc} can be applied to the MD sampler with reinitialization
(Remark~\ref{rmk:diagnostics}).
The update of $\bmomega^t$, up to translation by a scalar, and the update of $\bmphi^t$ 
correspond to, respectively, the update of $\bfW^t$ \eqref{eq:updatew} and
the update of $\bfV_k^t$ \eqref{eq:updateV} for any $k$ in the MD sampler.
We state four conditions for establishing the main results.
\begin{enumerate}
\item[(C1)]{The sample space $\calX$ is compact, $p(\bfx) >0$ for all $\bfx\in \calX$, $\Theta$ is bounded,
and $\Theta^*$ is nonempty.
The map $\bfg$ and the function $h$ are $p$-integrable and bounded.}
\item[(C2)]{There exist $\delta_q >0$ and $\epsilon_q >0$ such that $|\bfx-\bfy| \leq \delta_q$ implies that
$q(\bfx,\bfy) \geq \epsilon_q$ for all $\bfx,\bfy \in \calX$.}
\item[(C3)]{There exist an integer $\ell$, $\delta>0$ and a probability measure $\nu$, such that
$\nu(\calX_i)>0$ for $i=1,\ldots,\kappa$ and
$S^{\ell}_0(\bfx,A) \geq \delta \nu (A)$, $\forall \bfx \in \calX$ and $A \in \calB(\calX)$.
}
\item[(C4)]{For all $\bfx,\bfy\in\calX$ and all $\bmphi \in \Phi$, 
$t_{\bmphi}(\bfx,\bfy)>0$ and $\log t_{\bmphi}(\bfx,\bfy)$ has continuous partial derivatives with
respect to all the components of $\bmphi$.}
\end{enumerate}

To avoid mathematical complexity, we assume that $\calX$ is compact (C1). This assumption does not lose
much generality in practice as we may always restrict the sample space to 
$\{\bfx: p(\bfx)\geq \epsilon_p \}$ given an sufficiently small $\epsilon_p$. Due to the compactness of $\calX$,
any continuous map and function on $\calX$ will be bounded. Conditions (C2, C3) 
are standard conditions on the fixed proposal $q(\bfx,\bfy)$ to guarantee irreducibility and aperiodicity of 
the MH kernel $K_0$.
They are satisfied by all the local moves used in this article. A regularity
condition on the adaptive proposal $t_{\bmphi}$ is specified in (C4). For the mixed jumps in the examples, $\bmphi$ is either 
the covariance matrix of a multivariate normal distribution or the cell probability vector
of a multinomial distribution, and (C4) is satisfied.

\begin{lemma}\label{lemma1}
Let $\alpha \in [0,1)$. For any $i,j \in \{1,\ldots,\kappa\}$ and $\bmtheta \in \Theta$, 
if $e^{\omega_i-\omega_j}\geq c_1\in(0,1]$ then $K_{\bmtheta}(\bfx,A) \geq (1-\alpha) c_1 S_0(\bfx,A)$,
$\forall \bfx \in \calX_i$ and $A \in \calB(\calX_j)$.
\end{lemma}
\begin{proof}
By definition, $K_{\bmtheta}(\bfx,A) \geq (1-\alpha) K_{0,\bmomega}(\bfx,A) \geq (1-\alpha) S_{0,\bmomega}(\bfx,A)$ 
for every $\bmtheta=(\bmomega,\bmphi)$. For any $\bfx\in \calX_i$ and $\bfy\in \calX_j$, 
\begin{eqnarray*}
S_{0,\bmomega}(\bfx,d\bfy)  & = & q(\bfx,d\bfy) \min\left[1,e^{\omega_i-\omega_j} \frac{p(\bfy)q(\bfy,\bfx)}{p(\bfx)q(\bfx,\bfy)}\right] \\
 					& \geq &c_1 q(\bfx,d\bfy) \min\left[1, \frac{p(\bfy)q(\bfy,\bfx)}{p(\bfx)q(\bfx,\bfy)}\right] = c_1S_0(\bfx,d\bfy).
\end{eqnarray*}
Thus, $K_{\bmtheta}(\bfx,A) \geq (1-\alpha) c_1 S_0(\bfx,A)$ for all $A \in \calB(\calX_j)$.
\end{proof}

\begin{theorem}\label{thm:convergence}
If {\em (C1)--(C4)} hold, then 
$d(\bmtheta^n,\Theta^*)\toas 0$ and
\begin{equation}\label{eq:SLLN}
\frac{1}{n}\sum_{t=1}^n h(\bfX^t) \toas \mu_h{(\bmomega^*)}, \mbox{ as $n\to \infty$}.
\end{equation}
\end{theorem}

\begin{proof}
Lemma~\ref{lemma1} and (C3) implies that $\forall \bfx\in \calX_i$, 
$K^{\ell}_{\bmtheta}(\bfx,\calX_j) \geq \epsilon_{\nu} \nu(\calX_j)>0$
if $e^{\omega_i-\omega_j}\geq c_1$, where $\epsilon_{\nu}>0$. By Theorem 4.2 of Atchad\'e and Liu (2010), 
\begin{equation*}
\max_{i,j} \limsup_{n\to \infty} |v^n_i - v^n_j | < \infty, \;a.s.,
\end{equation*}
where $v^n_i = \sum_{t=1}^n \I (\bfX^t \in \calX_i)$. This
implies that the $\{\gamma_t\}$ defined by the MWL update (Routine~\ref{MWLdesign}) will
decrease below any given $\epsilon_{\gamma}>0$ after a finite number of iterations, i.e., $t_c<\infty$, 
almost surely.
Then, Algorithm~\ref{alg:damcmc} becomes a stochastic
approximation algorithm with a deterministic sequence of $\{\gamma_t\}$. According to Proposition 6.1 and
Theorem 5.5 in Andrieu et al. (2005), we only need to verify the drift conditions (DRI1-3) and assumptions (A1, A4)
given in that paper to show the convergence of $\bmtheta^n$.

{\em Verifying the drift conditions}.
Let $\calD$ be any compact subset of $\Theta$, and $\calD_1$ and $\calD_2$ be the projections
of $\calD$ into $\Omega$ and $\Phi$, respectively. Since $\calD_1$ is compact,
there is an $\epsilon_{\calD}\in (0,1]$ such that 
$\min_{i,j} \inf_{\bmomega \in \calD_1} e^{\omega_i-\omega_j} \geq \epsilon_{\calD}$.
By Lemma~\ref{lemma1} and (C3), there is a $\delta_{\calD}>0$ such that
\begin{equation}\label{eq:Kminor}
\inf_{\bmtheta \in \calD} K_{\bmtheta}^{\ell}(\bfx,A) \geq  \delta_{\calD} \nu(A),\, \forall \bfx \in \calX, A \in \calB(\calX),
\end{equation}
where $\ell$ and $\nu$ are defined in (C3). This gives the minorization condition in (DRI1).
Given (C2) and that $p_{\bmomega}, \bmomega \in \calD_1$, is bounded away from 0 and  $\infty$ under (C1), $K_{0,\bmomega}$
is irreducible and aperiodic for every $\bmomega \in \calD_1$, according to Theorem 2.2 of Roberts and Tweedie (1996). 
Consequently, for every $\bmtheta \in \calD$, $K_{\bmtheta}$ is also irreducible and aperiodic as $\alpha <1$.
Let $V(\bfx)=1$ for all $\bfx\in \calX$.
It is then easy to verify other conditions in (DRI1).

Since both $\bfg$ and $\Theta$ are bounded (C1), there is a $c_2>0$ such that for all $\bfx\in\calX$,
\begin{equation}\label{eq:supH}
\sup_{\bmtheta \in \Theta} \left|\bfH(\bmtheta,\bfx) \right| \leq \kappa + |\bfg(\bfx)| + \sup_{\bmphi\in\Phi}|\bmphi| \leq c_2
\end{equation}
\begin{equation}
\left|\bfH(\bmtheta,\bfx)-\bfH(\bmtheta',\bfx) \right|\leq |\bmphi-\bmphi'| \leq c_2 |\bmtheta-\bmtheta'|, \forall\, \bmtheta,\bmtheta' \in \Theta.
\end{equation}
These two inequalities imply (DRI2) with $V(\bfx)=1$.

Condition (DRI3) can be verified by the same argument used in Liang et al. (2007) once we
find a constant $c_3>0$ such that 
\begin{equation}\label{eq:spartial}
\left|\frac{\partial S_{\bmtheta}(\bfx,\bfy)}{\partial \theta_i} \right| \leq c_3 Q_{\bmphi}(\bfx,\bfy), 
\end{equation}
for all $\bfx,\bfy \in \calX$, $\bmtheta\in\calD$ and all $i$, where $\theta_i$ is the $i$th component
of $\bmtheta=(\bmomega,\bmphi)$. 
Denote by $\phi_j$ the $j$th component of $\bmphi$. Straightforward calculation leads to
$\left|{\partial S_{\bmtheta}(\bfx,\bfy)}/{\partial \omega_i} \right| \leq Q_{\bmphi}(\bfx,\bfy)$ and
\begin{equation*}
\frac{\partial S_{\bmtheta}(\bfx,\bfy)}{\partial \phi_j} =
\left\{\begin{array}{ll}
R_{\bmtheta}(\bfx,\bfy) \left[{\partial \log t_{\bmphi}(\bfy,\bfx)}/{\partial \phi_j}\right] \alpha t_{\bmphi}(\bfx,\bfy) , & \mbox{ if } R_{\bmtheta}(\bfx,\bfy) < 1 \\
\left[{\partial \log t_{\bmphi}(\bfx,\bfy)}/{\partial \phi_j}\right] \alpha t_{\bmphi}(\bfx,\bfy) , & \mbox{ otherwise,}
\end{array}\right.
\end{equation*}
where $R_{\bmtheta}(\bfx,\bfy) = p_{\bmomega}(\bfy)t_{\bmphi}(\bfy,\bfx)/p_{\bmomega}(\bfx)t_{\bmphi}(\bfx,\bfy)$. 
Condition (C4) with the compactness of $\calD_2$ and $\calX$ guarantees that
\begin{equation*}
\sup_{\bfx,\bfy\in \calX}\sup_{\bmphi \in \calD_2}\left|{\partial \log t_{\bmphi}(\bfx,\bfy)}/{\partial \phi_j}\right|<\infty.
\end{equation*} 
As $ \alpha t_{\bmphi}(\bfx,\bfy) \leq Q_{\bmphi}(\bfx,\bfy)$, \eqref{eq:spartial} holds and (DRI3) is verified.

{\em Verifying assumptions} (A1, A4). 
It is assumed in assumption (A1) the existence of a global Lyapunov function for $\bfF(\bmtheta)$.
Let $L(\bmtheta)=\frac{c_4}{2}\sum_{i=1}^{\kappa} (Z_i(\omega_i)-\bar{Z}(\bmomega))^2 + \frac{1}{2}|\bmphi-\bmmu_{\bfg,U}(\bmomega)|^2$, 
where $\bar{Z}(\bmomega)=Z(\bmomega)/\kappa$. Using straightforward algebra one can show that
\begin{equation*}
-Z \langle\nabla L, \bfF \rangle = 
\sum_{i=1}^{\kappa} c_4 Z_i\cdot(\Delta Z_i)^2 + 
(\Delta \bmphi)^{\mathsf{T}}\sum_{u\in U}\Delta Z_u \frac{\partial\bmmu_{\bfg,U}(\bmomega)}{\partial \omega_u}
+ \sum_{u\in U}Z_u |\Delta \bmphi|^2,
\end{equation*}
where $\Delta Z_i=Z_i(\omega_i)-\bar{Z}(\bmomega)$, $\Delta \bmphi=\bmphi-\bmmu_{\bfg,U}(\bmomega)$,
and the arguments ($\omega_i$ and $\bmomega$) in $Z_i(\omega_i)$ and $Z(\bmomega)$ have been dropped.
Since $\Theta$ is bounded, $Z_i(\omega_i) > \epsilon_{\Omega}>0$ for all $i$. Because $\bfg$ is $p$-integrable,
$\bmmu_{\bfg,i} \defi \bmmu_{\bfg,\{i\}}(\mathbf{0})$ is bounded for all $i$ and
\begin{equation*}
\left|\frac{\partial\bmmu_{\bfg,U}(\bmomega)}{\partial \omega_u} \right| \leq
|\bmmu_{\bfg,u}| + \max_{i\in U}|\bmmu_{\bfg,i}| \leq 2 \max_{1\leq i \leq \kappa} |\bmmu_{\bfg,i}|<\infty.
\end{equation*}
Thus, choosing a sufficiently large $c_4$ ensures that 
$\langle\nabla L(\bmtheta), \bfF(\bmtheta) \rangle \leq 0$ for any $\bmtheta \in \Theta$
with equality if and only if $\bmtheta\in\Theta^*$. 
Furthermore, $\{\bmtheta \in \Theta: L(\bmtheta)\leq C_L\}$ is compact for some $C_L >0$
and the closure of $L(\Theta^*)$ has an empty interior. Thus, all the conditions in assumption (A1) are satisfied. 
Since $|\bmtheta^{t+1}-\bmtheta^t|\leq \gamma_t \sup_{\bmtheta,\bfx} |\bfH(\bmtheta,\bfx)|\leq c_2\gamma_t$ \eqref{eq:supH} and
$\gamma_t = 1/(t+\xi)$ for $t>t_c$, verifying assumption (A4) is immediate.
This completes the proof of the convergence of $\bmtheta^n$.

The result \eqref{eq:SLLN} can be established similarly as the proof of Proposition 6.2 in 
Atchad\'e and Liu (2010). We only give an outline here. 
The drift conditions imply that for any $\bmtheta \in \calD$, there exist $h_{\bmtheta}(\bfx)$,
$c_5>0$, and $b \in (0,1]$
such that $h_{\bmtheta} - K_{\bmtheta} h_{\bmtheta} = h- \mu_h(\bmomega)$ and 
\begin{eqnarray*}
& \sup_{\bmtheta \in \calD} (\|h_{\bmtheta} \| + \| K_{\bmtheta}h_{\bmtheta}\|) < \infty,   \\
 & \|h_{\bmtheta}-h_{\bmtheta'}\|+\|K_{\bmtheta} h_{\bmtheta}-K_{\bmtheta'} h_{\bmtheta'}\|< c_5 |\bmtheta-\bmtheta'|^{b}, 
\forall\, \bmtheta,\bmtheta'\in \calD,
\end{eqnarray*}
where $K_{\bmtheta} h_{\bmtheta}(\bfx) = \int_{\calX} K_{\bmtheta}(\bfx,d\bfy) h_{\bmtheta}(\bfy)$ and 
for $f:\calX\to \R$, $\| f\| = \sup_{\bfx\in\calX} |f(\bfx)|$. See Proposition 6.1 and assumption (A3) of Andrieu et al. (2005).
Then, following an essentially identical proof to that of Lemma 6.6 in Atchad\'e and Liu (2010), we can show that
$\sum_{t=1}^{\infty}t^{-1}[h(\bfX^{t+1})-\mu_h(\bmomega^t)]$ has a finite limit almost surely.
Since $\mu_h(\bmomega^t) \toas \mu_h(\bmomega^*)$ as $t\to \infty$,
Kronecker's lemma applied to the above infinite sum leads to the desired result.

\end{proof}

\singlespacing


\begin{thebibliography}{99}

\bibitem{Andrieu05} Andrieu, C., Moulines, E., and Priouret, P. (2005),
``Stability of stochastic approximation under verifiable conditions," {\it SIAM Journal of Control
and Optimization}, 44, 283-312.

\bibitem{Andrieu06} Andrieu, C. and Moulines, E. (2006), ``On the ergodicity properties of some
adaptive MCMC algorithms," {\it Annals of Applied Probability}, 16, 1462-1505.

\bibitem{Atchade10} Atchad\'e, Y.F. and Liu, J.S. (2010), ``The Wang-Landau algorithm in general
state spaces: applications and convergence analysis," {\it Statistica Sinica}, 20, 209-233.

\bibitem{Becker97} Becker, O.M. and Karplus, M. (1997), ``The topology of multidimensional 
potential energy surface: Theory and application to peptide structure and kinetics,"
{\it Journal of Chemical Physics}, 106, 1495-1517.

\bibitem{Chen01} Chen, M.H., Shao, Q. and Ibrahim, J.G. (2001).
{\it Monte Carlo Methods in Bayesian Computation}. Springer, New York.

\bibitem{Cooper92} Cooper, G. F., and Herskovits, E. (1992), ``A Bayesian method for the induction 
of probabilistic networks from data," {\it Machine Learning}, 9, 309-347. 

\bibitem{Cooper99)} Cooper, G. F., and Yoo, C. (1999), ``Causal discovery from a mixture of 
experimental and observational Data," In {\it Proceedings of the 15th Conference on 
Uncertainty in Artificial Intelligence}, San Francisco, CA: Morgan Kaufmann, pp. 116-125.

\bibitem{Ellis08} Ellis, B and Wong, W. H. (2008), ``Learning causal Bayesian network structures
from experimental data," {\it Journal of the American Statistical Association}, 103, 778-789.

\bibitem{Friedman03} Friedman, N. and Koller, D. (2003), ``Being Bayesian about network
structure: a Bayesian approach to structure discovery  in Bayesian networks," 
{\it Machine Learning}, 50, 95-126.

\bibitem{Gelfand90} Gelfand, A. E. and Smith, A. F. M. (1990), ``Sampling-based
approaches to calculating marginal densities," {\it Journal of the American Statistical Association},
85, 398-409.

\bibitem{Geman84} Geman S. and Geman D. (1984), ``Stochastic relaxation, Gibbs
distributions and the Bayesian restoration of images,"
{\it IEEE Transactions on Pattern Analysis and Machine Intelligence}, 6, 721-741.

\bibitem{Geyer91} Geyer, C.J. (1991), ``Markov chain Monte Carlo maximum likelihood,"
In {\it Computing Science and Statistics: Proc. 23rd Symp. Interface}, E.M. Keramidas ed.,
Fairfax Station, VA, pp. 156-163.

\bibitem{Gilks94} Gilks, W.R., Roberts, R.O., and George, E.I. (1994), ``Adaptive direction
sampling," {\it The Statistician}, 43, 179-189.

\bibitem{Gordon93} Gordon, V.S. and Whitley, D. (1993), ``Serial and parallel genetic algorithms
as function optimizers," {\it Proc. 5th Int. Conf. Genetic Algorithms}, San Mateo, CA, 177-183.

\bibitem{Haario01} Haario, H., Saksman, E., and Tamminen, J. (2001), ``An adaptive Metropolis algorithm,"
{\it Bernoulli}, 7, 223-242.

\bibitem{Hastings70} Hastings, W.K. (1970), ``Monte Carlo sampling methods using Markov
chains and their applications," {\it Biometrika}, 57, 97-109.

\bibitem{Hindley07} Hindley, A. and Kolch, W. (2007), ``Raf-1 and B-Raf promote protein kinase C $\theta$
interaction with BAD," {\it Cellular Signalling}, 19, 547-555.

\bibitem{Holland88} Holland, P. W. (1988), ``Causal inference, path analysis, and recursive  structural
equations models," In {\it Sociological Methodology}, C.C. Clogg ed., Washington DC: American
Sociological Association, pp. 449-484.

\bibitem{Li87} Li, Z. and Scheraga, H.A. (1987), ``Monte Carlo-minimization approach to the multiple-minima
problem in protein folding," {\it Proceedings of the National Academy of Sciences, USA}, 84, 6611-6615.

\bibitem{Liang05} Liang, F (2005), ``Generalized Wang-Landau algorithm for Monte Carlo computation,"
{\it Journal of the American Statistical Association}, 100, 1311-1327.

\bibitem{Liang07} Liang, F., Liu, C., and Carroll, J. (2007), ``Stochastic approximation
in Monte Carlo computation," {\it Journal of the American Statistical Association}, 102, 305-320.

\bibitem{Liang09} Liang, F. (2009), ``On the use of stochastic approximation Monte Carlo for Monte Carlo
integration," {\it Statistics and Probability Letters}, 79, 581-587.

\bibitem{LiangZhang09} Liang, F and Zhang, J. (2009), ``Learning Bayesian networks for discrete data,"
{\it Computational Statistics and Data Analysis}, 53, 865-876. 

\bibitem{Liu01} Liu, J.S. (2008). {\it Monte Carlo Strategies in Scientific Computing} (2nd edition),
Springer, New York.

\bibitem{Liu00} Liu, J.S., Liang, F., and Wong, W.H. (2000), ``The multiple-try method and local optimization
in Metropolis sampling," {\it Journal of the American Statistical Association}, 95, 121-134.

\bibitem{Metropolis53} Metropolis, N., Rosenbluth, A.W., Rosenbluth, M. N., Teller, A.
H. and Teller, E. (1953), ``Equations of state calculation by fast computing
machines," {\it Journal of Chemical Physics}, 21, 1087-1091.

\bibitem{Neyman90} Neyman, J. (1990), ``Sur les applications de la thar des probabilities aux
experience Agaricales: Essay des principle," [English translation of excerpts by D. Dabrowska and T. Speed],
{\it Statistical Science}, 5, 463-472.

\bibitem{Pearl93} Pearl, J. (1993), ``Graphical models: causality and intervention," {\it Statistical 
Science}, 8, 266-273. 

\bibitem{Pearl00} Pearl, J. (2000), ``The logic of counterfactuals in causal inference,"
{\it Journal of the American Statistical Association}, 95, 428-435.

\bibitem{Roberts96} Roberts, G.O. and Tweedie, R.L. (1996), ``Geometric convergence and central limit
theorems for multidimensional Hastings and Metropolis algorithms," {\it Biometrika}, 83, 95-110.

\bibitem{Robins86} Robins, J. (1986), ``A new approach to causal inference in mortality studies 
with sustained exposure periods: Application to control of the healthy 
worker survivor effect," {\it Math Modeling}, 7, 1393-1512.

\bibitem{Robins87}  Robins, J. (1987), ``A graphical approach to the identification and estimation of 
causal parameters in mortality studies with sustained exposure periods,"
{\it Journal of Chronic Diseases}, 40, 1395-1615. 

\bibitem{Rubin78} Rubin, D. (1978), ``Bayesian inference for causal effects: The role of randomization,"
{\it Annals of Statistics}, 6, 34-58.

\bibitem{Sachs05} Sachs, K., Perez, O., Peer, D., Lauffenburger, D. A., and Nolan, G. P. (2005),
``Causal protein-signalling networks derived from multiparameter single-cell data,"
{\it Science}, 308, 523-529.

\bibitem{Spirtes} Spirtes, P. Glymour, C. and Scheines, R. (1993) {\it Causation, Prediction, and Search} (2nd ed.),
Cambridge, MA: MIT Press.

\bibitem{Tanner87} Tanner, M.A. and Wong, W.H. (1987), ``The calculation of posterior
distributions by data augmentation," {\it Journal of the American Statistical Association},
82, 528--550.

\bibitem{Wang01} Wang, F. and Landau, D.P. (2001), ``Efficient multiple-range random-walk
algorithm to calculate the density of states," {\it Physical Review Letters}, 86, 2050-2053.

\bibitem{Wales06} Wales, D.J. and Bogdan, T.V. (2006), ``Potential energy and free 
energy landscapes," {\it Journal of Physical Chemistry B}, 110, 20765-20776.

\bibitem{Wales97} Wales, D.J. and Doye, J.P.K. (1997), ``Global optimization by basin-hopping and
the lowest energy structures of Lennard-Jones clusters containing up to 110 atoms,"
{\it Journal of Physical Chemistry A}, 101, 5111-5116.

\bibitem{Zhou08} Zhou, Q. and Wong, W.H. (2008),
``Reconstructing the energy landscape of a distribution from Monte Carlo samples,"
{\it Annals of Applied Statistics}, 2, 1307-1331.

\bibitem{Zhou11} Zhou, Q. (2011), ``Random walk over basins of attraction to construct Ising energy landscapes,"
{\it Physical Review Letters}, 106, 180602.

\end{thebibliography}
\end{document}